\PassOptionsToPackage{table}{xcolor}
\documentclass[sigconf]{acmart}

\usepackage{enumerate}
\usepackage{booktabs}
\usepackage{epsfig}
\usepackage{graphicx}
\usepackage{xspace}
\usepackage{float}
\usepackage{subfigure,amsmath,latexsym}
\usepackage{tikz}
\usepackage{url}
\usepackage{color}
\usepackage{caption}
\usepackage{xcolor,colortbl}
\usepackage{multirow}
\usepackage{balance}

\newcommand{\cal}{\mathcal}

\usepackage[
resetcount,
algoruled,
linesnumbered,
vlined
]{algorithm2e}

\SetKw{OR}{or}
\SetKw{AND}{and}
\SetKw{NOT}{not}

\SetKw{TRUE}{true}
\SetKw{FALSE}{false}
\SetKw{NULL}{nil}

\SetKw{KwDownto}{downto}

\SetKwIF{gIf}{gElsIf}{gElse}{if}{then}{else if}{else}{end if}%
\SetKwIF{If}{ElseIf}{Else}{if}{then}{else if}{else}{end if}%
\SetKwFor{For}{for}{do}{end for}%
\SetKwFor{ForPar}{for}{do in parallel}{end for}%
\SetKwFor{ForEach}{for each}{do}{end for}%
\SetKwFor{ForAll}{forall}{do}{end for}%
\SetKwFor{While}{while}{do}{end while}%
\SetKw{Break}{break\;}%

\newcommand{\State}[1]{#1\;}
\newcommand{\StateCmt}[2]{
  #1; \tcc*[f]{#2}\;
}

\SetKwInOut{Input}{Input}
\SetKwInOut{Output}{Output}
\SetKwInput{KwDescription}{Description}

\renewcommand{\Return}[1]{\State{\textbf{return} #1}}

\let\oldnl\nl
\newcommand{\nonl}{\renewcommand{\nl}{\let\nl\oldnl}}

\newcommand{\stitle}[1]{\vspace{1ex} \noindent{\bf #1}}

\newcommand{\kw}[1]{{\ensuremath {\mathsf{#1}}}\xspace}

\newcommand{\la}{\leftarrow}

\newcommand{\cfig}{Figure~}
\newcommand{\ctab}{Table~}
\newcommand{\csec}{Section~}

\newcommand{\calg}{Algorithm~}
\newcommand{\cthm}{Theorem~}
\newcommand{\clem}{Lemma~}

\newcommand{\cequ}[1]{Equation~(#1)}

\newcommand{\ie}{{i.e.}}
\newcommand{\eg}{{e.g.}}

\newcommand{\bigo}{ {\cal O}}

\renewcommand{\widebar}[1]{\mkern 1.5mu\overline{\mkern-1.5mu#1\mkern-1.5mu}\mkern 1.5mu}

\newcommand{\kdbb}{\kw{KDBB}}
\newcommand{\madec}{\kw{MADEC^+}}
\newcommand{\madecp}{\kw{MADEC_p^+}}
\newcommand{\kdcwob}{\kw{kDC/UB1}}
\newcommand{\kdcwor}{\kw{kDC/RR3\&4}}
\newcommand{\kdcwoi}{\kw{kDC\text{-}Degen}}
\newcommand{\kdcwobr}{\kw{kDC/UB1\&RR3\&4}}
\newcommand{\mcbrb}{\kw{MC\text{-}BRB}}

\newcommand{\degenopt}{\kw{Degen\text{-}opt}}

\newcommand{\bbsearch}{\kw{Branch\&Bound}}
\newcommand{\bbsearcht}{\kw{Branch\&Bound\text{-}t}}

\newcommand{\lb}{lb}

\newcommand{\notS}{\widebar{S'}}
\newcommand{\gT}{\cal T}
\newcommand{\w}{\kw{w}}

\newcommand{\kdc}{\kw{kDC}}
\newcommand{\kdct}{\kw{kDC\text{-}t}}

\newcommand{\highlight}[1]{{#1}}

\renewcommand{\baselinestretch}{0.99}

\AtBeginDocument{%
  }

\begin{document}

\title{Efficient Maximum $k$-Defective Clique Computation with Improved
Time Complexity}

\author{Lijun Chang}
\email{Lijun.Chang@sydney.edu.au}
\orcid{0000-0002-6830-3900}
\affiliation{
    \institution{The University of Sydney}
	\city{Sydney}
    \country{Australia}
}

\setcopyright{acmlicensed}
\acmJournal{PACMMOD}
\acmYear{2023} \acmVolume{1} \acmNumber{3 (SIGMOD)} \acmArticle{209} \acmMonth{9} \acmPrice{15.00}\acmDOI{10.1145/3617313}

\begin{abstract}
$k$-defective cliques relax cliques by allowing up-to $k$ missing edges
from being a complete graph. This relaxation enables us to find larger
near-cliques and has applications in link prediction, cluster
detection, social network analysis and transportation science. The
problem of finding the largest $k$-defective clique has been recently
studied with several algorithms being proposed in the literature. 
However, the currently fastest algorithm \kdbb does not improve
its time complexity from being the trivial $\bigo(2^n)$, and also,
\kdbb's practical performance is still not satisfactory.
In this paper, we advance the state of the art for exact maximum
$k$-defective clique computation, in terms of both time
complexity and practical performance.
Moreover, we separate the techniques required for achieving the time
complexity from others purely used for practical performance
consideration; this design choice may facilitate the research
community to further improve the practical efficiency while not
sacrificing the worst case time complexity.
In specific, we first develop a general framework \kdc that beats the
trivial time complexity of $\bigo(2^n)$ and achieves a better time
complexity than all existing algorithms. The time complexity of \kdc is
solely achieved by our newly designed non-fully-adjacent-first branching rule, excess-removal
reduction rule and high-degree reduction rule.
Then, to make \kdc practically efficient, we further propose a new upper
bound, two new reduction rules, and an algorithm for efficiently
computing a large initial solution.
Extensive empirical studies on three benchmark graph
collections with $290$ graphs in total demonstrate that \kdc
outperforms the currently
fastest algorithm \kdbb by several orders of magnitude.
\end{abstract}


\maketitle  

\section{Introduction}
\label{sec:introduction}

The relationship among entities in many applications, such as social
media, communication networks, collaboration networks, web graphs, and
the Internet, can be naturally captured by the graph model. As a result,
real-world graph data is abundant, and graph-based data analysis has
been widely used to extract insights for guiding the decision-making
process.
In particular, the problem of identifying dense (\ie, cohesive)
subgraphs has been extensively studied~\cite{Book18:Chang,MMGD10:Lee},
since it serves many applications.
For example, identifying large dense subgraphs has been used for
detecting anomalies in financial networks~\cite{ahmed2016survey}, 
identifying real-time stories in social media~\cite{VLDBJ14:Angel},
detecting communities in social networks~\cite{bedi2016community},
and finding protein complexes in biological
networks~\cite{suratanee2014characterizing}.

Clique (\ie, complete subgraph) is a classic notion for defining dense
subgraphs, which requires every pair of distinct vertices in the
subgraph to be directly connected by an edge.
It is easy to see that a clique is the densest structure that a subgraph
can be. As a result, clique related problems have been extensively
explored in the literature, and many advancements have been made regarding
clique computation.
For example, it has been shown that the maximum clique is not only
NP-hard to compute exactly~\cite{CCC72:Karp}, but also NP-hard to
approximate within a factor of $n^{1-\epsilon}$ for any constant $0 <
\epsilon < 1$~\cite{FOCS96:Hastad}; here $n$ denotes the number of
vertices in the input graph $G$.
Nevertheless, exact algorithms have been studied both theoretically and
practically in the literature.
The state-of-the-art time complexity for maximum clique
computation is $\bigo^*(1.1888^n)$~\cite{clique}, and one of the
practically efficient algorithms is \mcbrb~\cite{KDD19:Chang}; here the
$\bigo^*$ notation hides polynomial factors.
In addition, the problems of enumerating all maximal cliques,
\highlight{enumerating all cliques of the maximum size}, and enumerating
and counting all cliques with $k$ vertices for a small $k$ have also been extensively
studied~\cite{VLDBJ20:Chang,JEA13:Eppstein,WALCOM17:Tomita,PVLDB20:Li,WSDM20:Jain}.

Requiring a large subgraph to be fully connected however is often too
restrictive for many applications, such as complex network
analysis~\cite{EOR13:Pattillo}, considering that data is often noisy or
incomplete.
Hence, various clique relaxations have been formulated in the
literature, such as quasi-clique~\cite{LATIN02:Abello},
$k$-plex~\cite{OR11:Balasundaram}, $k$-club~\cite{EJOR02:Bourjolly}, and
$k$-defective clique~\cite{Bio06:Yu}.
In this paper, we focus on the $k$-defective clique, which
allows a subgraph to miss up-to $k$ edges to be a complete subgraph;
note that a $k$-defective clique for $k=0$ is a clique. 
The concept of $k$-defective clique was formulated in~\cite{Bio06:Yu}
for predicting missing
interactions between proteins in biological networks. Besides, it also
finds applications in cluster detection~\cite{MP22:Stozhkov}, transportation
science~\cite{TS02:Sherali}, and social network
analysis~\cite{WWW20:Jain,IJC21:Gschwind}.
Since a clique is also a $k$-defective clique for any $k \geq 0$, the
maximum $k$-defective clique is no less than and usually can be much
larger than the maximum clique.
Consider the graph in \cfig\ref{fig:introduction}, it is easy to see
that the maximum clique size is $4$, while the maximum
$k$-defective clique size for any $k \leq 4$ is $4+k$; specifically, the
entire graph is a $4$-defective clique, and the remaining graph after
removing any vertex is a $3$-defective clique.

\begin{figure}[ht]
\centering
\includegraphics[scale=1.2]{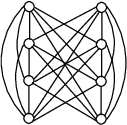}
\caption{Clique vs. $k$-Defective Clique}
\label{fig:introduction}
\end{figure}

%

The problem of maximum $k$-defective clique computation is
also NP-hard~\cite{STOC78:Yannakakis}.
The state-of-the-art time complexity for maximum $k$-defective clique
computation that beats the trivial $\bigo^*(2^n)$ time complexity is
achieved by the \madec algorithm proposed
in~\cite{COR21:Chen}, which runs in $\bigo^*(\sigma_k^n)$ time where
$\sigma_k < 2$ is the largest real root of the equation $x^{2k+3} -
2x^{2k+2}+1=0$. Although a graph coloring-based upper bound as well as
other pruning techniques are proposed in~\cite{COR21:Chen} aiming to improve
the practical performance of \madec, it is shown in~\cite{AAAI22:Gao}
that the graph coloring-based upper bound proposed in~\cite{AAAI22:Gao}
is ineffective and \madec is inefficient in practice especially when $k \geq
10$. For example, for $k \geq 15$ on the Facebook graphs collection
(please refer to \csec\ref{sec:experiment} for the description of the
dataset), even the version of \madec that is further optimized by the
authors of~\cite{AAAI22:Gao} was still not able to find the maximum
$k$-defective clique for any graph instance with a time limit of $3$ hours.
With the goal of enhancing the practical performance, the \kdbb algorithm is
designed in~\cite{AAAI22:Gao} which proposes and incorporates
preprocessing as well as multiple pruning techniques. Nevertheless, 
\kdbb is still inefficient, and moreover, no
time complexity better than $\bigo^*(2^n)$ has been proved for \kdbb.

In this paper, we aim to advance the state of the art for maximum
$k$-defective clique computation, both theoretically and practically.
We first develop a general backtracking framework \kdc based on our newly
designed \highlight{non-fully-adjacent-first} branching rule ({\bf BR}),
\highlight{excess-removal reduction rule ({\bf RR1}) and high-degree
reduction rule ({\bf RR2})}. We prove that our framework runs in $\bigo^*(\gamma_k^n)$ time
where $\gamma_k < 2$ is the largest real root of the equation
$x^{k+3}-2x^{k+2}+1=0$. 
In comparison, the time complexity of \madec is $\bigo^*(\gamma_{2k}^n)$
by observing that $\sigma_k = \gamma_{2k}$.
Note that $\gamma_k < \gamma_{2k}$.
Thus, we advance the state of the art regarding the theoretical time
complexity.
We remark that the time complexity of \kdc is solely achieved by our
branching rule {\bf BR} and reduction rules {\bf RR1} and {\bf RR2}, and
these are the minimal requirements for achieving the time
complexity of $\bigo^*(\gamma_k^n)$.
We deliberately separate the techniques required for
achieving the time
complexity from the ones used purely for improving the practical
performance, such that others may further improve the efficiency while
retaining the time complexity of $\bigo^*(\gamma_k^n)$.

To make \kdc practically efficient, we further propose techniques from
three aspects: \highlight{an improved graph coloring-based upper bound
({\bf UB1}), a degree-sequence-based reduction rule ({\bf RR3}), a
second-order reduction rule ({\bf RR4})},
and a new algorithm \degenopt for efficiently
computing a large initial solution.
Specifically, given a graph $g$ and a $k$-defective clique represented
by its set of vertices $S$ such that $S \subseteq V(g)$, \highlight{our
improved coloring-based upper bound} {\bf UB1}
computes an upper bound of the largest $k$-defective clique that is in
$g$ and contains $S$, and prunes the backtracking instance $(g,S)$ if the
computed upper bound is no larger than the currently found largest
solution.
Same as the graph coloring-based upper bound proposed
in~\cite{COR21:Chen}, {\bf UB1} also utilizes graph coloring; but {\bf
UB1} computes a much tighter (\ie, smaller) upper bound
than~\cite{COR21:Chen}.
In essence, a (greedy) graph coloring is used to partition the vertices
into independent sets, by observing that all vertices with the same
color form an independent set.
Let $\pi_1,\ldots,\pi_c$ be a partitioning of $V(g)\setminus S$ into
independent sets. The upper bound computed in~\cite{COR21:Chen} is $|S|
+ \sum_{i=1}^c \min(\lfloor \frac{1+\sqrt{8k+1}}{2}\rfloor, |\pi_i|)$,
which is based on the observation that an independent set with more than
$\lfloor \frac{1+\sqrt{8k+1}}{2}\rfloor$ vertices will miss more than
$k$ edges and thus cannot be all contained in the same $k$-defective
clique. The upper bound of~\cite{COR21:Chen} has two deficiencies.
Firstly, it allows $c$ independent sets, each of size up-to $\lfloor
\frac{1+\sqrt{8k+1}}{2}\rfloor$, to be included into the solution for
computing the upper bound; this will actually introduce almost $c\times
k$ missing edges, much larger than the allowed $k$ missing edges. In
particular, if $|\pi_i| \geq \lfloor \frac{1+\sqrt{8k+1}}{2}\rfloor$ for
each $1 \leq i \leq c$, then the upper bound becomes $|S| + c\times
\lfloor \frac{1+\sqrt{8k+1}}{2}\rfloor$, while $|S|+c+k$ is a much
smaller upper bound.
Secondly, it ignores the missing edges within $S$ and the missing edges
between $S$ and $V(g)\setminus S$; for example, the upper bound remains
the same even if $S$ already has $k$ missing edges.
Our {\bf UB1} computes a tighter upper bound than
$|S|+c+k-|\overline{E}(S)|$ by resolving the above two issues; here
$\overline{E}(S)$ denotes the set of missing edges in $S$.

\stitle{Contributions.}
Our main contributions are as follows.
\begin{itemize}
	\item We develop a general framework \kdc for computing the maximum
		$k$-defective clique in $\bigo^*(\gamma_k^n)$ time, based on our
		newly designed branching rule {\bf BR} and reduction rules {\bf
		RR1} and {\bf RR2}; here $\gamma_k < 2$ is the largest real root
		of the equation $x^{k+3}-2x^{k+2}+1=0$. (\csec\ref{sec:framework})
	\item We propose a new upper bound {\bf UB1} based on graph
		coloring, which can be computed in linear time and is much
		tighter than the upper bounds proposed
		in~\cite{COR21:Chen,AAAI22:Gao}.
		(\csec\ref{sec:upper_bound})
	\item We propose two new reduction rules {\bf RR3} and {\bf RR4}
		that can be conducted in linear time.
		(\csec\ref{sec:reduction_rules})
	\item We propose an algorithm \degenopt for computing a large
		initial $k$-defective clique in $\bigo(\delta(G)\times m)$ time,
		where $m$ is the number of edges and $\delta(G) \leq \sqrt{m}$
		is the degeneracy of $G$. (\csec\ref{sec:initial})
\end{itemize}
We also conduct extensive empirical studies on \highlight{three}
benchmark graph collections with \highlight{$290$} graph instances in total to
evaluate our techniques (\csec\ref{sec:experiment}).
The results show that (1)~on the real-world graphs collection, \kdc with a
time limit of $3$ {\em seconds} solves even more graph instances than the
existing fastest algorithm \kdbb with a time limit of $3$ {\em hours}, and
(2)~on the $41$ Facebook graphs that have more than $15,000$ vertices,
\kdc is on average three orders of magnitude faster than \kdbb.
In addition, our ablation studies demonstrate that each of our
additional
techniques (\ie, upper bound {\bf UB1}, reduction rules {\bf RR3} and
{\bf RR4}, and initial solution computation) improves the
practical efficiency of \kdc.

\section{Preliminaries}
\label{sec:preliminaries}

In this paper, we focus on a large {\em unweighted} and {\em undirected} 
graph $G=(V, E)$, where $V$ is the set of vertices and $E$ is the set of
undirected edges; we consider only {\em simple} graphs, \ie, without
self-loops and parallel edges.
Let $n = |V|$ and $m=|E|$ denote the cardinalities of $V$ and $E$,
respectively.
%
%
We denote the undirected edge between $u$ and $v$ by both $(u,v)$ and
$(v,u)$; then, $u$ (resp. $v$) is said to be adjacent to and a neighbor
of $v$ (resp. $u$).
The set of neighbors of $u$ in $G$ is $N_G(u) = \{v \in V \mid (u,v) \in
E\}$, and the {\em degree} of $u$ in $G$ is $d_G(u) = |N_G(u)|$.
Given a vertex subset $S$ of $G$, we use $G[S]$ to denote the subgraph
of $G$ induced by $S$, \ie, $G[S] = (S, \{(u,v) \in E \mid u,v \in
S\})$.
For ease of presentation, we simply refer to an unweighted and undirected
graph as a graph, and omit the subscript $G$ from the notations when the
context is clear.
For an arbitrary given graph $g$, we denote its set of vertices and its
set of edges by $V(g)$ and $E(g)$, respectively. 

\begin{definition}[Clique]
	A graph $g$ is a clique (\ie, complete graph) if it has an edge
	between every pair of distinct vertices, \ie, $|E(g)| =
	\frac{|V(g)|(|V(g)|-1)}{2}$ or equivalently, 
	$d_g(u) = |V(g)|-1,\forall u \in V(g)$.
\end{definition}

\begin{definition}[$k$-Defective Clique]
	A graph $g$ is a $k$-defective clique if it misses at most $k$
	edges, \ie, $|E(g)| \geq \frac{|V(g)|(|V(g)|-1)}{2} - k$.
\end{definition}

The definition of $k$-defective clique relaxes the definition of clique
by allowing a few (\ie, $k$) missing edges, and $0$-defective cliques are
cliques.
Obviously, if a subgraph $g$ of $G$ is a $k$-defective clique, then the
subgraph of $G$ induced by vertices $V(g)$ is also a
$k$-defective clique. Thus, in this paper, \textbf{\em we simply refer
to a $k$-defective clique by its set of vertices}, and measure the size
of a $k$-defective clique $C \subseteq V$ by the number of vertices,
\ie, $|C|$.

\begin{figure}[htb]
\centering
\begin{tikzpicture}[scale=0.5]
\node[draw, circle, inner sep=1.5pt] (1) at (0,0) {$v_1$};
\node[draw, circle, inner sep=1.5pt] (2) at (0,-2) {$v_2$};
\node[draw, circle, inner sep=1.5pt] (3) at (2,-2) {$v_3$};
\node[draw, circle, inner sep=1.5pt] (4) at (4,-2) {$v_4$};
\node[draw, circle, inner sep=1.5pt] (5) at (4,0) {$v_5$};
\node[draw, circle, inner sep=1.5pt] (6) at (2,0) {$v_6$};
\node[draw, circle, inner sep=1.5pt] (7) at (2,2) {$v_7$};
\node[draw, circle, inner sep=1.5pt] (8) at (6,0) {$v_8$};
\node[draw, circle, inner sep=1.5pt] (9) at (6,-2) {$v_9$};
\node[draw, circle, inner sep=0pt] (10) at (9,-2) {$v_{10}$};
\node[draw, circle, inner sep=0pt] (11) at (9,0) {$v_{11}$};
\node[draw, circle, inner sep=0pt] (12) at (7.5,2) {$v_{12}$};
\path[draw,thick] (1) -- (2) -- (3) -- (4) -- (5) -- (6) -- (1) -- (7)
-- (5);
\path[draw,thick] (1) -- (3) -- (5) -- (2) -- (6) -- (4) -- (1);
\path[draw,thick] (3) -- (6) -- (7);
\path[draw,thick] (8) -- (9) -- (10) -- (11) -- (12) -- (8) -- (10) --
(12) -- (9) -- (11) -- (8);
\end{tikzpicture}
\caption{An example graph}
\label{fig:graph}
\end{figure}
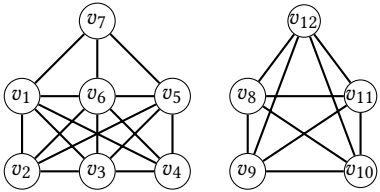

The property of $k$-defective clique is {\em hereditary}, \ie, any
subset of a $k$-defective clique is also a $k$-defective clique.
A $k$-defective clique $C$ of $G$ is a {\em maximal $k$-defective
clique} if every proper superset of $C$ in $G$ is not a $k$-defective
clique, and is a {\em maximum $k$-defective clique} if its size is the
largest among all $k$-defective cliques of $G$; note that the maximum
$k$-defective clique is not unique.
Consider the graph in \cfig\ref{fig:graph}, $\{v_8,v_9,\ldots,v_{12}\}$
is a maximum clique and is also a maximum $1$-defective clique. In
addition, both $\{v_1,v_2,v_3,v_4,v_6\}$ and $\{v_1,v_2,v_3,v_5,v_6\}$
are maximum $1$-defective cliques that miss the edge $(v_2,v_4)$ and the
edge $(v_1,v_5)$, respectively. $\{v_1,v_2,\ldots,v_6\}$ is a maximum
$2$-defective clique that misses edges $(v_2,v_4)$ and $(v_1,v_5)$.

To facilitate the presentation, we denote the set of edges that are
missing from a graph $g$ by $\overline{E}(g)$, \ie, $(u,v) \in
\overline{E}(g)$ if and only if $u \neq v$ and $(u,v) \notin E(g)$; we
call the edges of $\overline{E}(g)$ as \textbf{\em non-edges} of
$g$.
Thus, $g$ is a $k$-defective clique if and only if $|\overline{E}(g)|
\leq k$.
For two vertices $u$ and $v$ that are not adjacent (\ie, not connected
by an edge), we call $v$ (resp. $u$) a \textbf{\em non-neighbor}
of $u$ (resp. $v$); note that \textbf{\em a vertex is considered neither
a neighbor nor a non-neighbor of itself}.
We denote the set of all non-neighbors of $u$ in $G$ by
$\overline{N}_G(u) = V(G)\setminus (N_G(u)\cup u)$; note that, for
presentation simplicity, we denote the union of a set $S$ and a vertex
$u$ by $S\cup u$, and denote the subtraction of $u$ from $S$ by
$S\setminus u$.
For any set $S$ of vertices and a vertex $u$ (where $u$ could be either
in or not in $S$), we abbreviate $N_{G[S\cup u]}(u)$ as $N_S(u)$ and
abbreviate $\overline{N}_{G[S\cup u]}(u)$ as $\overline{N}_S(u)$.

\stitle{Problem Statement.} Given a graph $G = (V,E)$ and an integer $k
\geq 1$, we study the problem of maximum $k$-defective
clique computation, aiming to find the largest $k$-defective clique
in $G$.

Frequently used notations are summarized in \ctab\ref{table:notations}.

\begin{table}[t]
\small
\caption{Frequently used notations}
\label{table:notations}
\begin{tabular}{cp{.7\columnwidth}} \hline
	Notation & Meaning \\ \hline
	$G = (V,E)$ & an unweighted and undirected graph with vertex set $V$
	and edge set $E$ \\
	$g = (V(g), E(g))$ & a subgraph of $G$ \\
	$S, C \subseteq V$ & $k$-defective cliques \\
	$N_S(u)$ & the set of $u$'s {\em neighbors} that are in $S$ \\
	$\overline{N}_S(u)$ & the set of $u$'s {\em non-neighbors} that are in $S$ \\
	$d_S(u)$ & the {\em number} of $u$'s neighbors that are in $S$ \\
	$E(S)$ & the set of (undirected) {\em edges} in the subgraph of $g$ (or
	$G$) induced by $S$ \\
	$\overline{E}(S)$ & the set of (undirected) {\em non-edges} in the
	subgraph of $g$ (or $G$) induced by $S$ \\
	\hline
\end{tabular}
\end{table}

\subsection{Degeneracy Ordering, $k$-Core and $k$-Truss}

In this subsection, we review the concepts of degeneracy ordering,
\highlight{$k$-core and $k$-truss, which will be used in
\csec\ref{sec:time_complexity}}.

\begin{definition}[Degeneracy ordering]
\label{definition:degeneracy_ordering}
Given a graph $G$, an ordering $(v_1,v_2,\ldots,v_n)$ of its
vertices is a degeneracy ordering if for each $1 \leq i \leq n$, $v_i$
is the vertex with the smallest degree in the subgraph of $G$ induced by
vertices $\{v_i,v_{i+1},\ldots,v_n\}$.
\end{definition}

\begin{definition}[$k$-core~\cite{SN83:Seidman}]
\label{definition:k_core}
Given a graph $G$ and an integer $k$, the $k$-core of $G$ is the maximal
subgraph $g$ of $G$ such that every vertex $u \in V(g)$ has degree
$d_g(u) \geq k$ in the subgraph $g$.
\end{definition}

$k$-core is a {\em vertex-induced} subgraph.
The degeneracy ordering can be computed in $\bigo(m)$ time by the
peeling algorithm~\cite{JACM83:Matula,Book18:Chang}, which iteratively
removes the vertex with {\em the smallest degree} from the graph and appends
it to the end of the ordering. Note that, although the $k$-core can also
be computed by the peeling algorithm, it is usually more efficient to
directly compute the $k$-core by iteratively removing vertices of
{\em degree smaller than $k$} from the graph~\cite{Book18:Chang}.
The largest $k$ such that $G$ contains a non-empty $k$-core is known as the
\textbf{\em degeneracy} of $G$, denoted $\delta(G)$; note that
$\delta(G) \leq \sqrt{m}$~\cite{JEA13:Eppstein}.
For the graph in \cfig\ref{fig:graph}, $(v_7, v_1, v_2, v_3, v_4, v_5,
v_6, v_8, v_9, v_{10}, v_{11}, v_{12})$ is a degeneracy ordering.
The entire graph is a $3$-core, and the subgraph obtained by removing
$v_7$ is a $4$-core; $\delta(G) = 4$, as it has no $5$-core.

\begin{definition}[$k$-truss~\cite{PVLDB12:Wang}]
\label{definition:k_truss}
Given a graph $G$ and an integer $k$, the $k$-truss of $G$ is the
maximal subgraph $g$ of $G$ such that every edge $(u,v) \in E(g)$
participates in at least $k-2$ triangles, \ie, 
$|N_g(u)\cap~N_g(v)| \geq k-2, \forall (u,v) \in E(g)$.
\end{definition}

$k$-truss is a subgraph of the $(k-1)$-core, and is an {\em
edge-induced} subgraph.
$k$-truss can be considered as a higher-order version of
$k$-core. That is, each edge corresponds to a node, and each triangle
corresponds to a hyper-edge, in a hyper-graph. Hence, the $k$-truss can
be computed in a similar way to $k$-core, but the time complexity becomes
$\bigo(\delta(G)\times m)$~\cite{PVLDB12:Wang}.
For the graph in \cfig\ref{fig:graph}, the entire graph is a $3$-truss,
the subgraph obtained by removing edges $\{(v_7,v_1), (v_7,v_6),
(v_7,v_5)\}$ (and thus also vertex $v_7$) is a $4$-truss, and the
subgraph induced by vertices $\{v_8,v_9,\ldots,v_{12}\}$ is a $5$-truss
which is contained in the $4$-core.

\section{Our Approach}
\label{sec:approach}

In this section, we propose an efficient algorithm \kdc for exact
maximum $k$-defective clique computation. As the maximum $k$-defective
clique computation problem is NP-hard~\cite{STOC78:Yannakakis}, our
algorithm \kdc, as well as all other exact algorithms, will run in
exponential time in the worst case. Nevertheless, our algorithm \kdc
beats the trivial time complexity of $\bigo^*(2^n)$ where the $\bigo^*$
notation hides polynomial factors. Specifically, we prove that our
algorithm \kdc runs in $\bigo^*(\gamma_k^n)$ time where $\gamma_k < 2$
is the largest real root of the equation $x^{k+3} - 2x^{k+2}+1=0$;
note that this improves the state-of-the-art time complexity
$\bigo^*(\gamma_{2k}^n)$~\cite{COR21:Chen}, as
$\gamma_k$ increases regarding $k$ (\ie, $\gamma_k < \gamma_{2k}$).

In the following, we first in \csec\ref{sec:framework} present the
framework of \kdc and prove its time complexity. Then, we in
\csec\ref{sec:upper_bounds} propose upper bounds and
reduction rules to improve the practical performance of \kdc.
Lastly, we in \csec\ref{sec:initial} present a heuristic algorithm for
initially computing a large $k$-defective clique.

\subsection{The Framework of \kdc}
\label{sec:framework}

Our algorithm falls into the category of branch-and-bound search (also
known as backtracking) algorithms;
%
we will use the terms {\em backtracking} and {\em branch-and-bound
search} interchangeably. The general idea is as follows.
Let $(g,S)$ denote an instance of the backtracking, where $g$ is
a graph and $S \subseteq V(g)$ is a $k$-defective clique in $g$.
The goal of solving an instance is to find the largest $k$-defective
clique in the instance; {\em a $k$-defective clique is said to be
in the instance $(g,S)$ if it is in $g$ and contains $S$}. To solve
the instance $(g,S)$, a backtracking algorithm will
select a branching vertex $b \in V(g)\setminus S$, and then recursively
solve two new instances that are generated based on $b$: one
instance includes $b$ into $S$, and the other removes $b$ from $g$
(and thus excludes $b$ from being added into $S$). 
Solving the instance $(G,\emptyset)$ thus finds the maximum $k$-defective
clique in $G$.
All the instances that are generated in solving the instance
$(G,\emptyset)$ form a binary search tree, a snippet of which
is shown in \cfig\ref{fig:search_tree}. 
Each node of the search tree, denoted by $I_i$, represents an
instance of the backtracking (\ie, $I_i = (g_i, S_i)$), and has two
children representing the two new instances that are generated based on
the branching vertex of the instance $I_i$.
For example, in \cfig\ref{fig:search_tree}, the branching vertex
selected for the instance $I_0$ is $b_0$, and the two new instances that are
generated based on $b_0$ are $I_1$ (which includes $b_0$ into the
solution) and $I_{q+1}$ (which removes $b_0$ from the graph); the
actions of including $b_0$ and removing $b_0$ are, respectively,
represented as labels $+b_0$ and $-b_0$ on the corresponding edges in
the search tree.

\begin{figure}[t]
\centering
\includegraphics[scale=0.3]{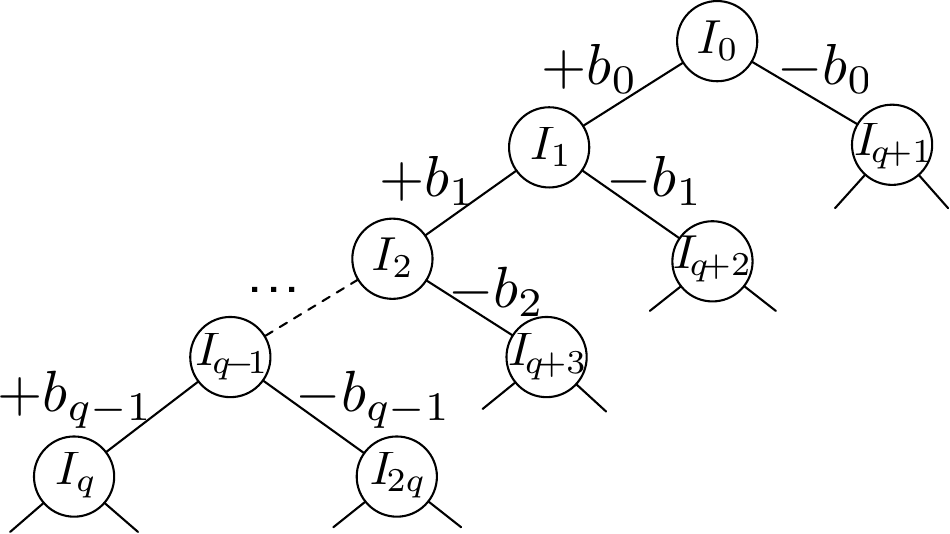}
\caption{A snippet of the (binary) search tree $\gT$ of a backtracking
algorithm}
\label{fig:search_tree}
\end{figure}

Backtracking algorithms differ from each other in three aspects:
\vspace{-2pt}
\begin{itemize}
	\item {\bf Branching techniques} determine which vertex is selected
		as the branching vertex, \eg, $b_0$ for the instance $I_0$ and $b_1$
		for the instance $I_1$ in \cfig\ref{fig:search_tree}.
	\item {\bf Reducing techniques} reduce the size of an instance, \ie,
		transform an instance $(g,S)$ to another equivalent instance
		$(g',S')$ with $|V(g')\setminus S'| \leq |V(g) \setminus
		S|$.\footnote{Note that, reducing techniques could also
			remove edges from the graph, \eg,
			the reduction rule {\bf RR6} in
			\csec\ref{sec:reduction_rules}.
			We omit the discussions here for
		simplicity.}
	\item {\bf Upper bounding techniques} prune an instance, as well as
		the entire search subtree rooted at the instance, if a computed
		upper bound of the largest $k$-defective clique in the instance
		is no larger than the best solution found so far.
\end{itemize}

In this paper, we propose new techniques from all the three aspects for
the problem of maximum $k$-defective clique computation. In this
subsection, we only present the techniques that are required to achieve our
time complexity of $\bigo^*(\gamma_k^n)$, and defer other practical
techniques to Sections~\ref{sec:upper_bounds} and \ref{sec:initial}.

\subsubsection{Techniques for Achieving Our Time Complexity}

We first propose the following
\highlight{non-fully-adjacent-first branching rule (\textbf{BR}), which
	prefers branching on a vertex that is not fully
adjacent to $S$}.
\begin{description}
	\item[BR \highlight{(non-fully-adjacent first branching rule)}.] Given an instance $(g,S)$, the branching vertex is
		selected as the one of $V(g)\setminus S$ that has at least one
		non-neighbor in $S$; if all vertices of $V(g)\setminus S$ are
		adjacent to all vertices of $S$, then the branching vertex is an
		arbitrary vertex of $V(g)\setminus S$.
\end{description}
Note that, the way of selecting a branching vertex will not compromise
the correctness of the algorithm, {\em as long as the union of $S$ and the
branching vertex forms a $k$-defective clique}.
%
To achieve our time complexity, we also propose the
following two reduction rules.
\begin{description}
	\item[RR1 \highlight{(excess-removal reduction rule)}.] Given an instance $(g,S)$, for a vertex $u \in
		V(g)\setminus S$ satisfying $|\overline{E}(S\cup u)| > k$, we
		remove $u$ from $g$.
	\item[RR2 \highlight{(high-degree reduction rule)}.] Given an instance $(g,S)$, for a vertex $u \in
		V(g)\setminus S$ satisfying $|\overline{E}(S\cup u)| \leq k$ and
		$d_g(u) \geq |V(g)|-2$, we greedily add $u$ to $S$.
\end{description}
The reduction rule {\bf RR1} ensures that the union of $S$ and any
branching vertex (including the one selected by our branching rule {\bf
BR}) form a valid
$k$-defective clique, by noting that {\em all reduction rules are
applied before the branching rule}.
The correctness of the reduction rule {\bf RR1} is trivial, and we prove
the correctness for the reduction rule {\bf RR2} in the lemma below.

\begin{lemma}
Given an instance $(g,S)$, for a vertex $u \in V(g)\setminus S$
satisfying $|\overline{E}(S\cup u)| \leq k$ and $d_g(u) \geq |V(g)|-2$,
there is a maximum $k$-defective clique in the instance that contains $u$.
\end{lemma}

\begin{proof}
The case of $d_g(u) = |V(g)|-1$ (\ie, $u$ is adjacent to all other
vertices in $g$) is trivial.
Let's focus on the case of $d_g(u) = |V(g)|-2$ and consider a maximum
$k$-defective clique $C$ in the instance that does not contain $u$, \ie,
$u \notin C$ and $S \subseteq C \subseteq V(g)$. Let $v$ be the unique
non-neighbor of $u$ in $g$. Then, $v$ must be in $C$, as otherwise
$C\cup u$ would be a $k$-defective clique of size larger than $C$. We
consider two cases depending on whether $v \in S$.

\noindent
\underline{Case-I: $v \notin S$.} That is, $v \in C\setminus S$. It is
easy to verify that $C\cup u \setminus v$ is a $k$-defective clique of
the same size as $C$ and contains $u$ and $S$.

\noindent
\underline{Case-II: $v \in S$.}
There must exist a vertex of $C\setminus S$ that has at least one
non-neighbor in $C$, since otherwise $C\cup \{u\}$ would also be a valid
$k$-defective clique by noting that $|\overline{E}(S\cup u)| \leq k$;
let $w$ be such a vertex of $C\setminus S$. It is easy to verify
that $C\cup u \setminus w$, which contains $u$ and $S$, is a
$k$-defective clique of the same size as $C$.
\end{proof}

\begin{algorithm}[htb]
\small
\caption{$\kdct(G, k)$}
\label{alg:kdct}
\KwIn{A graph $G$ and an integer $k$}
\KwOut{A maximum $k$-defective clique $C^*$ of $G$}

\vspace{1.5pt}
\State{$C^* \la \emptyset$}
\State{$\bbsearcht(G,\emptyset)$}
\Return{$C^*$}

\vspace{1.5pt}
\nonl \textbf{Procedure} $\bbsearcht(g,S)$ \\
\State{$(g',S') \la $ apply reduction rules {\bf RR1} and {\bf RR2} to $(g,S)$}
\lIf{$g'$ is a $k$-defective clique}{update $C^*$ by $V(g')$ and
\textbf{return}}
	\StateCmt{$b\la$ a vertex of $V(g')\setminus S'$ that has at least
		one non-neighbor in $S'$}{\footnotesize If there is no such a vertex, then
	$b$ is an arbitrary vertex of $V(g')\setminus S'$}
	\StateCmt{$\bbsearcht(g', S'\cup b)$}{\footnotesize Left branch includes $b$}
	\StateCmt{$\bbsearcht(g'\setminus b, S')$}{\footnotesize Right branch excludes
	$b$}
\end{algorithm}

\highlight{Based on the above discussions, the pseudocode of our 
algorithm \kdct is shown in \calg\ref{alg:kdct}; here $\kw{t}$
stands for theoretical as the algorithm only considers the theoretical
aspect. It takes a graph $G$ and an integer $k$ as input, and outputs a
maximum $k$-defective clique $C^*$ of $G$ which is achieved by
recursively invoking \bbsearcht to grow a partial solution $S$ that is
initialized as $\emptyset$ (Line~2).
In the procedure $\bbsearcht$, we first apply reduction rules {\bf RR1}
and {\bf RR2} to reduce the instance $(g,S)$ to a potentially smaller
instance $(g',S')$ satisfying $V(g')\subseteq V(g)$ and $S' \supseteq S$
(Line~4). 
%
If $g'$ itself is a $k$-defective clique, then we update the currently
found largest $k$-defective clique $C^*$ by $V(g')$ and backtrack
(Line~5).
Otherwise, we pick a branching vertex $b$ based on our branching rule
\textbf{BR} (Line~6), and then generate two new instances of \bbsearcht
and go into recursion (Lines~7--8).
}

\begin{figure}[ht]
\centering
\includegraphics[scale=.8]{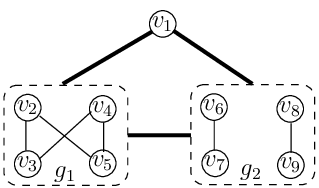}
\caption{\highlight{Running example for \calg\ref{alg:kdct} (thick edge
indicates full connection between the subgraphs)}}
\label{fig:kdct}
\end{figure}

\highlight{
\begin{example}
\label{example:kdct}
Consider the example graph in \cfig\ref{fig:kdct} where thick edge
indicates full connection between the corresponding subgraphs, \ie,
$v_1$ is adjacent to every other vertex, and every vertex of $g_1$ is
adjacent to every vertex of $g_2$. Suppose $k=3$, when invoking
\bbsearcht with $g=G$ and $S=\emptyset$, our reduction rule {\bf RR2}
will greedily and iteratively move $v_1,v_2,v_3,v_4,v_5$ to $S$. Then,
an arbitrary vertex of $\{v_6,\ldots,v_9\}$ can be selected as the
branching vertex; suppose $v_6$ is selected. The newly
generated left branch would have $S_1 = \{v_1,\ldots,v_6\}$, and the
reduction rules {\bf RR1} and {\bf RR2} would have no effect for $S_1$;
note that the graph $g$ will remain unchanged in the remaining part of
this example. The branching vertex selected for $S_1$ could be either
$v_8$ or $v_9$ (as they are not fully adjacent to $S_1$) but not $v_7$
(which is fully adjacent to $S_1$); suppose $v_8$ is selected. The newly
generated left branch for $S_1$ would have $S_2 =
\{v_1,\ldots,v_6,v_8\}$ which contains three non-edges, and thus the
reduction rule {\bf RR1} will remove $v_7$ and $v_9$ from the graph.
\end{example}
}

\subsubsection{Time Complexity Analysis of \calg\ref{alg:kdct}}

To analyze the time complexity of \calg\ref{alg:kdct}, we
consider the search tree $\gT$ of recursively invoking \bbsearcht, 
as shown in \cfig\ref{fig:search_tree}. To avoid confusion, we
refer to nodes of the search tree by {\em nodes}, and vertices of a
graph by {\em vertices}.
Recall that each node of $\gT$ represents an instance of \bbsearcht, \ie,
$(g,S)$, and has two children: the left child includes the branching
vertex $b$ to $S$, and the right child excludes $b$ from $g$. 
%
It is worth mentioning that each child may also include or exclude other
vertices due to applying reduction rules \textbf{RR1} and
\textbf{RR2}.
%
We use $I, I', I_0, I_1,\ldots$ to denote nodes of $\gT$, and use $I.g$
and $I.S$ to respectively denote the graph $g$ and the partial solution
$S$ of the \bbsearcht instance to which $I$ corresponds.
We would like to emphasize that \textbf{\em
$I.g$ and $I.S$ denote the ones obtained after applying the
reduction rules} at Lines~4--5 of \calg\ref{alg:kdct}, not the ones input
to \bbsearcht; note that Line~5 can be regarded as applying the
following reduction rule:
\begin{itemize}
	\item If $g'$ is a $k$-defective clique, then all vertices of
		$V(g')\setminus S'$ are moved to $S'$.
\end{itemize}
In this case, the instance will not generate any
children (\ie, any new instances) and thus becomes a leaf node.
We measure the size of $I$ by the number of vertices in the graph $I.g$
that are not in the partial solution $I.S$, \ie, $|I| = |V(I.g)|-|I.S| =
|V(I.g)\setminus I.S| \geq 0$. It is easy to see that {\em $|I'| \leq |I|-1$
whenever $I'$ is a child of $I$ --- \eg, the branching vertex $b$ of $I$ is in
$V(I.g)\setminus I.S$ but not in $V(I'.g)\setminus I'.S$ --- and $|I|=0$
whenever $I$ is a leaf node}.

Before proving the time complexity, we first state
the following important property of exhaustively applying the
reduction rules {\bf RR1} and {\bf RR2}, whose proof is omitted due to
space limitation.

\begin{lemma}
\label{lemma:property_RR1_RR2}
After exhaustively applying the reduction
rules \textbf{RR1} and \textbf{RR2}, the resulting instance $(g,S)$
satisfies the following condition:
\begin{itemize}
	\item For every vertex $u \in V(g)\setminus S$, it holds that
		$|\overline{E}(S\cup u)| \leq k$ and $d_g(u) < |V(g)|-2$.
\end{itemize}
\ie, \textbf{all vertices of $V(g)\setminus S$ have at least two
non-neighbors in $g$}.
\end{lemma}

We are now ready to prove the time complexity of \calg\ref{alg:kdct}
in the following lemma and theorem.

\begin{lemma}
	\label{lemma:leafs}
	Let $\gT$ be  the search tree of running \calg\ref{alg:kdct} (\ie,
	recursively invoking \bbsearcht). For any node $I$ of $\gT$, the
	number of leaf nodes in the subtree of $\gT$ rooted at $I$, denoted
	$\ell(I)$, is at most $\gamma_k^{|I|}$, where $1 <
	\gamma_k < 2$ is the largest real root of the equation
	$x^{k+3}-2x^{k+2}+1=0$.
\end{lemma}

\begin{proof}
We prove the lemma by induction.
For the base case that $I$ is a leaf node, it is trivial that $\ell(I) =
1 \leq \gamma_k^{|I|}$ since $\gamma_k > 1$ and $|I| = 0$.
For a non-leaf node $I$, for any path $(I_0=I,I_1,\ldots,I_{q-1},I_q)$
with $q \geq 1$ that starts from $I$ and always visits the left child in
the search tree $\gT$, it is trivial that 
\[
	\ell(I) = \ell(I_{q+1}) + \ell(I_{q+2}) + \cdots + \ell(I_{2q}) +
	\ell(I_q)
\]
here, $I_{q+1},I_{q+2},\ldots,I_{2q}$ are the right child of
$I_0,I_1,\ldots,I_{q-1}$, respectively, as illustrated in
\cfig\ref{fig:search_tree}.
To bound $\ell(I)$, let's specifically consider the path $(I_0=I,
I_1,\ldots,I_q)$ where $I_q$ is the first node such that $|I_q| \leq
|I_{q-1}|-2$; this implies that for $1 \leq i < q$, $|I_i| =
|I_{i-1}|-1$ and consequently $V(I_i.g) = V(I.g)$ and the reduction
rules at Line~4 of \calg\ref{alg:kdct} have no effect on $I_{i-1}$.
Note that such a node $I_q$ always exists since (1)~$I.g$ is not a
$k$-defective clique (otherwise, $I$ would be a leaf node) and (2)~a
leaf node $I'$ satisfies $|I'|=0$ (\ie, $I'.S = V(I'.g) \neq V(I.g)$ and
thus $I'$ would satisfy the condition).
Hence, the following two facts hold.
\begin{description}
	\item[Fact 1.] $|I_i| \leq |I_{i-q-1}|-1 \leq |I| + q-i$, for $q+1 \leq
		i \leq 2q$.
	\item[Fact 2.] $|I_q| \leq |I_{q-1}|-2 \leq |I|-q-1$.
\end{description}

Now, we prove that the following fact also holds.
\begin{description}
	\item[Fact 3.] $q \leq k+1$.
\end{description}
We prove Fact 3 by contradiction. Suppose $q \geq k+2$. Let $I_x$ be the
last node, on the path $(I_0 = I, I_1, \ldots,I_x,\ldots, I_q)$,
satisfying the condition that all vertices of $V(I_x.g)\setminus I_x.S$
are adjacent to all vertices of $I_x.S$, \ie, the branching vertex $b_x$
selected for $I_x$ has no non-neighbor in $I_x.S$; without loss of
generality, we assume such an $I_x$ exists, otherwise the following
proof still holds by setting $x = 0$.
Then, $|\overline{E}(I_x.S)| \geq x$ since (1)~each branching vertex added
to $I_x.S$ along the path $(I_0=I, \ldots, I_{x-1})$ has at least two
non-neighbors (according to \clem\ref{lemma:property_RR1_RR2}), and
(2)~all these non-neighbors are in $I_x.S$ (according to the definition
of $I_x$); note that, according to the definition of $I_q$, the
reduction rules have no effect on $I_i$ for $1 \leq i < q$. This implies
that $x \leq k$.
Then, according to the definition of $I_x$ and our branching rule
\textbf{BR}, for each $i$ with $x+1\leq i < q$, the branching vertex
selected for $I_i$ has at least one non-neighbor in
$I_i.S$, and consequently, 
\[
	|\overline{E}(I_q.S)| \geq |\overline{E}(I_x.S)| + (q-x-1) \geq q-1 \geq k+1
\] 
contradicting that
$I_q.S$ is a $k$-defective clique. Hence, Fact 3 holds.

Based on Facts 1, 2 and 3, we have 
\begin{align}
	\ell(I) & = \ell(I_{q+1}) + \ell(I_{q+2}) + \cdots + \ell(I_{2q}) +
	\ell(I_q) \nonumber \\
	& \leq \gamma_k^{|I_{q+1}|} + \gamma_k^{|I_{q+2}|} + \cdots +
	\gamma_k^{|I_{2q}|} + \gamma_k^{|I_q|} \nonumber \\
	& \leq \gamma_k^{|I|-1} + \gamma_k^{|I|-2} + \cdots +
	\gamma_k^{|I|-q} + \gamma_k^{|I|-q-1} \nonumber \\
	& \leq \gamma_k^{|I|-1} + \gamma_k^{|I|-2} + \cdots +
	\gamma_k^{|I|-k-1} + \gamma_k^{|I|-k-2}
	\label{eq:time_complexity_proof}
\end{align}
where $\gamma_k^{|I|-1} + \gamma_k^{|I|-2} + \cdots + \gamma_k^{|I|-k-1}
+ \gamma_k^{|I|-k-2} \leq \gamma_k^{|I|}$ if $\gamma_k$ is
no smaller than the largest real root of the equation $x^{k+2} -
x^{k+1} -\cdots - x - 1=0$ which is equivalent to the equation
$x^{k+3}-2x^{k+2}+1=0$~\cite{Book2010:Fomin}. The first few solutions to
the equation are $\gamma_0 = 1.619$, $\gamma_1 = 1.840$, $\gamma_2 = 1.928$,
$\gamma_3=1.966$, $\gamma_4=1.984$ and $\gamma_5=1.992$.
\end{proof}

\begin{theorem}
\label{theorem:time_complexity}
Let $P_1$ be the time complexity of running Lines~4--6 of
\calg\ref{alg:kdct}, which is polynomial in $n$. Then, the
time complexity of \calg\ref{alg:kdct} is $\bigo(P_1\times
\gamma_k^n)$ and is $\bigo^*(\gamma_k^n)$, where $\gamma_k < 2$ is the
largest real root of the equation $x^{k+3} - 2x^{k+2} + 1 = 0$.
\end{theorem}

\begin{proof}
Firstly, as the search tree $\gT$ is a full binary tree, the total number of
nodes in the search tree is at most twice the number of its leaf nodes.
Secondly, it is easy to see that 
the time complexity of each node is $P_1$. Thus, the theorem holds.
\end{proof}

The existing best time complexity for the problem
of maximum $k$-defective clique computation is achieved by the algorithm
\madec~\cite{COR21:Chen}, which is $\bigo^*(\sigma_k^n)$ where
$\sigma_k$ is
the largest real root of $x^{2k+3} - 2x^{2k+2}+1=0$. It is easy to see
that $\sigma_k = \gamma_{2k}$. Thus, the time complexity of \madec
is $\bigo^*(\gamma_{2k}^n)$ and is higher than
our time complexity considering that $\gamma_k < \gamma_{2k}$.
\highlight{
Our algorithm has two main features that enable the improved time
complexity. Firstly, after exhaustively applying our reduction rules
{\bf RR1} and {\bf RR2}, every vertex in $V(g)\setminus S$ will have at
least two non-neighbors as stated in \clem\ref{lemma:property_RR1_RR2};
as a result, we can prove {\bf Fact 3} above. If we do not apply {\bf
RR2}, then we will only be able to bound $q$ by $2k+1$ for {\bf Fact 3}
and get the same time complexity as~\cite{COR21:Chen}. For example,
consider the graph in \cfig\ref{fig:kdct} again and suppose $k=2$ and
$I_0.S = \{v_1\}$. The branching vertex selected for $I_0$ could be $b_0
= v_2$, and we have $I_1.S = \{v_1,v_2\}$. The branching vertex selected
for $I_1$ could be $b_1=v_4$ and we get $I_2.S=\{v_1,v_2,v_4\}$.
Similarly, we could have $b_2=v_3$, $I_3.S=\{v_1,\ldots,v_4\}$,
$b_3=v_5$, $I_4.S=\{v_1,\ldots,v_5\}$, $b_4=v_6$ and
$I_5.S=\{v_1,\ldots,v_6\}$. Only in instance $I_5$, the reduction rules
will finally have any effect (\eg, remove vertices $\{v_8,v_9\}$); thus
$q=5$.
The second feature of our algorithm is our new branching rule.
We remark that after incorporating our reduction rules {\bf RR1} and {\bf
RR2} into~\cite{COR21:Chen}, its time complexity could be
improved to roughly $\bigo^*(\gamma_{3k/2}^n)$. However, this is still
higher than our time complexity. Thus, our branching rule is better than
that of \madec~\cite{COR21:Chen}.

It is interesting to observe that the state-of-the-art time complexity
for maximum $k$-plex computation is similar to $\bigo^*(\gamma_k^n)$;
specifically, it is
$\bigo^*(\gamma_{k-1}^n)$~\cite{AAAI20:Zhou,CIKM22:Dai}.
This is because an
inequality similar to \cequ{\ref{eq:time_complexity_proof}} is also
proved and utilized in~\cite{AAAI20:Zhou,CIKM22:Dai}. But, we remark
that our techniques and arguments to obtain
\cequ{\ref{eq:time_complexity_proof}} are
different from that of~\cite{AAAI20:Zhou,CIKM22:Dai} due to different
problem natures.
}

\begin{algorithm}[htb]
\small
\caption{$\kdc(G, k)$}
\label{alg:kdc}

\State{$C^* \la$ heuristically compute a large $k$-defective clique of $G$}
\State{$g \la $ apply reduction rules to reduce $G$}
\State{$\bbsearch(g, \emptyset)$}
\Return{$C^*$}
\end{algorithm}

\subsubsection{Incorporating Other Practical Techniques.}
\label{sec:practical_kdc}

\highlight{
\calg\ref{alg:kdct} is used for illustrating the bare minimum needed to 
achieve our time complexity of $\bigo^*(\gamma_k^n)$, and its practical
performance would not be satisfactory. Thus, we propose to
further incorporate other practical techniques, such as preprocessing,
upper bound-based pruning and more reduction rules, into our algorithm.
The pseudocode of our practically improved algorithm \kdc
is given in \calg\ref{alg:kdc}.
%
In \kdc, we first heuristically compute a large
$k$-defective clique of $G$ (Line~1),
whose details will be given in \csec\ref{sec:initial}.
Let $C^*$ be the heuristically computed $k$-defective clique. We then
use $|C^*|$ to reduce $G$ by removing unpromising vertices
and edges (Line~2); the details will be given in
\csec\ref{sec:time_complexity}.
After that, we go into backtracking by invoking the procedure \bbsearch.

\bbsearch is similar to \bbsearcht in \calg\ref{alg:kdct}, but with the
following additions.
Firstly, besides {\bf RR1} and {\bf RR2}, we also apply
other reduction rules at Line~4 of \calg\ref{alg:kdct}; these reduction
rules will be presented in \csec\ref{sec:reduction_rules}.
Secondly, before picking a branching vertex at Line~6 of
\calg\ref{alg:kdct}, we also compute an upper bound of the maximum
$k$-defective clique in the instance $(g',S')$, and prune the entire instance
if the computed upper bound is no larger than $|C^*|$; this implies that
no $k$-defective
clique in the instance $(g',S')$ is of size larger than $C^*$. Details
of the upper bound computation will be presented in
\csec\ref{sec:upper_bound}.

Let $P_2$ be the time complexity of running Lines~1--2 of
\calg\ref{alg:kdc}. The time complexity of \kdc is $\bigo(P_2 +
P_1\times \gamma_k^n)$ and is also $\bigo^*(\gamma_k^n)$. That is, applying
these additional techniques does not affect the exponential part of our
time complexity in \cthm\ref{theorem:time_complexity}.
}

\subsection{Upper Bounds and Reduction Rules}
\label{sec:upper_bounds}



\subsubsection{Upper Bounds}
\label{sec:upper_bound}

For the upper bound-based pruning,
we propose an improved graph coloring-based upper bound.
Before that, we first present the existing graph coloring-based upper bound
proposed in~\cite{COR21:Chen}, where the graph coloring is mainly used
to partition the vertices into independent sets.
Specifically, a coloring of a graph is assigning each vertex a color
such that for every edge in the graph, its two end-points
have different colors. Given an instance $(g,S)$ and a coloring of
$V(g)\setminus S$ with $c$ distinct colors $\{1,2,\ldots,c\}$, let
$\pi_1,\pi_2,\ldots,\pi_c$ be the partitioning of
$V(g)\setminus S$ based on their colors; that is, each $\pi_i$ consists of
all vertices with color $i$ and thus is an independent set. The
upper bound in~\cite{COR21:Chen} is
\begin{equation}
	\label{eq:color_bound_1}
	\textstyle
	|S| + \sum_{i=1}^c \min\left(\left\lfloor
	\frac{1+\sqrt{8k+1}}{2}\right\rfloor, |\pi_i|\right)
\end{equation}
This is based on the observation that an independent set with more than
$\lfloor \frac{1+\sqrt{8k+1}}{2}\rfloor$ vertices will miss more than $k$
edges and thus cannot be all contained in the same $k$-defective clique.
However, the upper bound computed by \cequ{\ref{eq:color_bound_1}} has
two deficiencies.
\begin{itemize}
	\item It considers the partitions $\pi_1,\ldots,\pi_c$
		independently, and thus would include much more vertices than
		necessary for computing the upper bound. For example, suppose
		$|\pi_i| \geq \lfloor \frac{1+\sqrt{8k+1}}{2}\rfloor$, $\forall 1
		\leq i \leq c$, then the upper bound becomes $|S| + c \cdot
		\lfloor \frac{1+\sqrt{8k+1}}{2}\rfloor$. But obviously $|S| +
		c + k$ is a much smaller upper bound (\eg, when $c$
		is large); this is because adding any $s_i$ vertices of $\pi_i$
		to $S$ will introduce at least $s_i-1$ non-edges.
	\item It does not consider the non-edges in $S$, and the
		non-edges between $S$ and $V(g)\setminus S$.
\end{itemize}
As a result, the upper bound computed by \cequ{\ref{eq:color_bound_1}}
is not tight.

\begin{figure}[ht]
\centering
\includegraphics[scale=1]{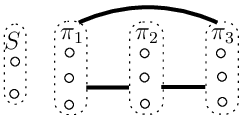}
\caption{Running example for graph coloring-based upper bound
	computation (the graph consists of $11$ vertices and $27$
		edges where $(\pi_1,\pi_2,\pi_3)$ forms a $3$-partite clique,
		\ie, all edges are between a vertex of $\pi_i$ and a vertex
of $\pi_j$ for $i \neq j$)}
\label{fig:upper_bound}
\end{figure}

\begin{example}
\label{example:upper_bound}
Consider the graph $g$ in \cfig\ref{fig:upper_bound} and the partial
solution $S$ consisting of two isolated vertices (\ie, without any
adjacent edges).
Besides $S$, the other part of $g$ is a $3$-partite clique with vertex sets
$\pi_1,\pi_2$ and $\pi_3$; note that, there is no edge between $S$ and
$V(g)\setminus S$. Thus, a graph coloring of $V(g)\setminus S$ would
assign all vertices of $\pi_i$ with color $i$, for $1 \leq i \leq 3$.
Suppose $k = 3$, then $\lfloor \frac{1+\sqrt{8k+1}}{2}\rfloor = 3$.
As $|\pi_1| = |\pi_2| = |\pi_3| = 3$, the graph coloring-based upper
bound computed by~\cequ{\ref{eq:color_bound_1}} is $2+3\times 3 = 11$.
However, it is easy to observe that the maximum $k$-defective clique in
the instance $(g,S)$ is of size only $3$, as we can only add one more vertex
without violating the $k$-defective clique definition.
\end{example}

In this paper, we still utilize the graph coloring-based partitioning
$\pi_1,\ldots,\pi_c$ for computing the upper bound, but we compute a
much tighter (\ie, smaller) upper bound than
\cequ{\ref{eq:color_bound_1}} by resolving the above two deficiencies,
as follows.
\begin{description}
	\item[UB1 \highlight{(improved coloring-based upper bound)}.] 
	For each partition $\pi_i$, we first sort its vertices into
	non-decreasing order regarding $|\overline{N}_S(\cdot)|$ and then
	define the weight of the \mbox{$j$-th} vertex in the sorted order, denoted
	$v_{i_j}$, to be $\w(v_{i_j}) = |\overline{N}_S(v_{i_j})| + j-1$,
	where the index $j$ starts from $1$.
	Finally, let $v_1,v_2,\ldots,$ be an ordering of $V(g)\setminus S$ in
	non-decreasing  order regarding their weights $\w(\cdot)$.
	The maximum $k$-defective clique in the instance $(g,S)$ is of size
	at most
	$|S|$ plus the largest $i$ such that $\sum_{j=1}^i \w(v_j) \leq k -
	|\overline{E}(S)|$.
\end{description}
The general idea is that (1)~adding any $s_i \geq 0$ vertices of $\pi_i$
to $S$ will introduce at least $\sum_{j=1}^{s_i} \w(v_{i_j}) \geq
\frac{s_i(s_i-1)}{2}$ non-edges, and (2)~$(s_1,s_2,\ldots,s_c)$ is
determined greedily and interdependently. It can be verified that the
upper bound computed by {\bf UB1} is at most (and can be much smaller
than) $|S| + c + k - |\overline{E}(S)|$, and is also no larger than that
computed by \cequ{\ref{eq:color_bound_1}}.

\begin{proof}[Proof of UB1]
For any $k$-defective clique $C$ in the instance $(g,S)$ that
contains $s_i$ vertices of $\pi_i$ for each $1 \leq i \leq c$, the
number of missing edges in $C$ is
\begin{align*}
	|\overline{E}(C)| & \textstyle \geq |\overline{E}(S)| + \sum_{i=1}^c
	\left( \frac{s_i(s_i-1)}{2} + \sum_{j=1}^{s_i}
	|\overline{N}_S(v_{i_j})|\right) \\
	& \textstyle = |\overline{E}(S)| + \sum_{i=1}^c \sum_{j=1}^{s_i}
	\left(|\overline{N}_S(v_{i_j})| + j-1\right) \\
	& \textstyle = |\overline{E}(S)| + \sum_{i=1}^c \sum_{j=1}^{s_i} \w(v_{i_j}) \\
	& \textstyle \geq |\overline{E}(S)| + \sum_{j=1}^{s_1+\cdots+s_c} \w(v_j)
\end{align*}
The first inequality follows from the fact that adding $s_i$ vertices
$S_i$ of $\pi_i$ to $S$ will introduce
\begin{itemize}
	\item at least $\frac{s_i(s_i-1)}{2}$ non-edges between vertices of
		$S_i$ (since $S_i$ is an independent set), and
	\item at least $\sum_{j=1}^{s_i} |\overline{N}_S(v_{i_j})|$
		non-edges between $S$ and $S_i$ (since the vertices of $\pi_i$
			are ordered such that $|\overline{N}_S(v_{i_1})| \leq
			|\overline{N}_S(v_{i_2})| \leq |\overline{N}_S(v_{i_3})| \leq
		\cdots$).
\end{itemize}
The second inequality follows from the fact that the vertices of
$V(g)\setminus S$ are ordered such that $\w(v_1) \leq \w(v_2) \leq
\w(v_3) \leq \cdots$.
Consequently, {\bf UB1} follows from the fact that $|\overline{E}(C)|
\leq k$ as $C$ is a $k$-defective clique.
\end{proof}

\begin{example}
Continue Example~\ref{example:upper_bound}. Now we show how {\bf
UB1} computes the upper bound. As $|\overline{N}_S(v)|
= 2$ for each $v \in V(g)\setminus S$, the weights of the vertices in
$\pi_i$ for $1 \leq i \leq 3$ are all $\{2,3,4\}$. Thus, the weights of
all vertices of $V(g)\setminus S$ are $\{2,2,2,3,3,3,4,4,4\}$. As
$|\overline{E}(S)| = 1$, the upper bound computed by {\bf UB1} is
$2+1=3$, which is significantly tighter than that computed in
Example~\ref{example:upper_bound}.
\end{example}

In addition, we also adopt the following two upper bounds from the
literature.
\begin{description}
	\item[UB2.] The maximum $k$-defective clique in the instance $(g,S)$
		is of size at most $\min_{u \in S} d_g(u) + 1 +
		k$~\cite{COR21:Chen}.
	\item[UB3.] Given an instance $(g,S)$, let $v_1,v_2,\ldots$ be an
		ordering of $V(g)\setminus S$ in non-decreasing order regarding
		their numbers of non-neighbors in $S$, \ie
		$|\overline{N}_S(\cdot)|$. The maximum $k$-defective
		clique in the instance $(g,S)$ is of size at most $|S|$ plus the
		largest $i$ such that $\sum_{j=1}^i |\overline{N}_S(v_j)| \leq k -
		|\overline{E}(S)|$~\cite{AAAI22:Gao}.
\end{description}

\subsubsection{Reduction Rules}
\label{sec:reduction_rules}

Besides the two reduction rules {\bf RR1} and {\bf RR2} presented in
\csec\ref{sec:framework}, we further propose two new reduction rules based
on the size of the currently found best solution (\ie, currently found
largest $k$-defective clique).
\textbf{\em Let $\lb$ be the size of the currently found best solution}, then
we will not be interested in any solution of size $\leq \lb$.
We first present and prove the reduction rule {\bf RR3} that is derived
from the upper bound {\bf UB3}.
\begin{description}
	\item[RR3 \highlight{(degree-sequence-based reduction rule)}.] Given an instance $(g,S)$, let $v_1,v_2,\ldots$ be an
		ordering of $V(g)\setminus S$ in non-decreasing order regarding
		their numbers of non-neighbors in $S$, \ie,
		$|\overline{N}_S(\cdot)|$. For a vertex $v_i$ with $i
> \lb - |S|$ and $|\overline{N}_S(v_i)| > k - |\overline{E}(S)| -
\sum_{j=1}^{\lb - |S|} |\overline{N}_S(v_j)|$, we can remove $v_i$ from $g$.
\end{description}

\begin{proof}[Proof of RR3]
Consider the instance $(g,S\cup v_i)$ which is obtained from $(g,S)$ by
adding $v_i \in V(g)\setminus S$ to $S$, and denote $S\cup v_i$ by $S'$.
Let $v'_1,v'_2,\ldots$ be an ordering of $V(g)\setminus S'$ in
non-decreasing order regarding their numbers of non-neighbors in $S'$.
{\bf UB3} states that the maximum $k$-defective clique in the instance
$(g,S')$ is of size at most $|S'|$ plus the largest $i'$ such that
$\sum_{j=1}^{i'} |\overline{N}_{S'}(v'_j)| \leq k - |\overline{E}(S')|$.
Note that, we have
\begin{align*}
	\textstyle \sum_{j=1}^{\lb-|S|} |\overline{N}_{S'}(v'_j)| &
	\textstyle \geq
	\sum_{j=1}^{\lb - |S|} |\overline{N}_S(v_j)| \\
	& > k - |\overline{E}(S)|
	- |\overline{N}_S(v_i)| 
	 = k - |\overline{E}(S')|
\end{align*}
where the first inequality follows from the fact that
$|\overline{N}_{S'}(v)| \geq~|\overline{N}_S(v)|$ holds for every $v \in
V(g)\setminus S'$, and the second inequality follows from the statement
of the reduction rule {\bf RR3}.
Consequently, the maximum $k$-defective clique in the instance $(g,S')$
is of size strictly less than $|S'| + (\lb - |S|) = \lb+1$. That is,
every $k$-defective
clique that is in the instance $(g,S)$ and contains $v_i$ is of size at most
$\lb$, and thus we can remove $v_i$ from $g$.
\end{proof}

Then, we propose the reduction rule {\bf RR4}.
\begin{description}
	\item[RR4 \highlight{(second-order reduction rule)}.] Given an instance $(g,S)$ with $S\neq \emptyset$, for
		any vertex $u \in S$ and $v \in V(g)\setminus S$, let $S' =
		S\cup v$, $cn(u,v)$ be the number of common neighbors of $u$ and
		$v$ in $\notS = V(g)\setminus~S'$, $cnon(u,v)$ be the number of
		common non-neighbors of $u$ and $v$ in $\notS$, $xn(u,v)$ 
		be the number of vertices that are exclusive neighbors of either
		$u$ or $v$ in $\notS$;
		specifically,
		$xn(u,v) = |N_{\notS}(u) \setminus N_{\notS}(v)| +
		|N_{\notS}(v) \setminus N_{\notS}(u)|$. \newline
		If $|S'| + cn(u,v) + \min\big(k - |\overline{E}(S')|,
		xn(u,v)\big) + \min\big(cnon(u,v)$,
			$\max (0,  \lfloor \frac{k - |\overline{E}(S')|
		- xn(u,v)}{2} \rfloor)\big) \leq \lb$, then we can
		remove $v$ from $g$.
\end{description}

\begin{proof}[Proof of RR4]
	Let's consider the instance $(g,S')$. It is easy to verify that
	$cn(u,v)$, $cnon(u,v)$ and $xn(u,v)$ represent disjoint subsets of
	vertices of $V(g)\setminus S'$ and $cn(u,v)+cnon(u,v)+xn(u,v) =
	|V(g)\setminus S'|$.
	We consider two cases depending on
	whether $k-|\overline{E}(S')| > xn(u,v)$. Firstly, if
	$k-|\overline{E}(S')| \leq xn(u,v)$, then the maximum $k$-defective
	clique in the instance $(g,S')$ will be of size at most
	$|S'|+cn(u,v)+(k-|\overline{E}(S')|)$, since (1)~adding any exclusive
	neighbor of $u$ or $v$ to $S'$ will introduce at least one
	non-edge and (2)~adding any common non-neighbor of $u$ and $v$ to
	$S'$ will introduce at least two non-edges. Similarly, if
	$k-|\overline{E}(S')| > xn(u,v)$, then the maximum $k$-defective
	clique in the instance $(g,S')$ will be of size at most
	$|S'|+cn(u,v)+xn(u,v)+\min(cnon(u,v), \lfloor
	\frac{k-|\overline{E}(S')|-xn(u,v)}{2}\rfloor)$. In summary, the
	maximum $k$-defective clique in the instance $(g,S')$ is of size at
	most $|S'| + cn(u,v) + \min\big(k - |\overline{E}(S')|,
		xn(u,v)\big) + \min\big(cnon(u,v)$, \\
			$\max (0,  \lfloor \frac{k - |\overline{E}(S')|
		- xn(u,v)}{2} \rfloor)\big)$ and {\bf RR4} is correct.
\end{proof}

From the proofs of {\bf RR3} and {\bf RR4}, it can be observed that the
reduction rules are actually designed based on upper bounds; for
example, {\bf RR3} is based on {\bf UB3}.
The general idea is that, given an instance $(g,S)$ and a vertex $v \in
V(g)\setminus S$, if an upper bound of the instance $(g,S\cup v)$ is at
most $\lb$, then we can remove $v$ from $g$. 
It is easy to see that an alternative strategy is to directly
generate the instance $(g,S\cup v)$ which will then be pruned by the
upper bounds; this will have the same pruning effects as the reduction
rules.
The advantage of using reduction rules to remove $v$ from $g$
is that the reduction rules can be applied more efficiently by
computation sharing; in particular, applying the reduction rules for all
vertices of $V(g)\setminus S$ can be conducted in linear time in total
(see \csec\ref{sec:time_complexity}), while generating all the
sub-instances and then pruning by the upper bounds would take quadratic
time.
On the other hand, it is also worth mentioning that an upper bound could
be designed based on {\bf RR4}; we do not use it in this paper since
computing this upper bound is time-consuming.

In addition, we also utilize the following two reduction rules from the
literature.
\begin{description}
	\item[RR5.] Given an instance $(g,S)$, for any vertex $v \in
		V(g)\setminus S$ whose degree is less than $\lb-k$, we
		can remove $v$ from $g$~\cite{COR21:Chen}.
	\item[RR6.] Given a graph $G$, for any edge $(u,v) \in E(G)$ whose
		number of common neighbors in $G$ is less than $\lb-k-1$, we can
		remove the edge $(u,v)$ from $G$~\cite{AAAI22:Gao}.
\end{description}

\subsubsection{Time Complexity Analysis}
\label{sec:time_complexity}

Now, we analyze the time complexity of all the upper bounds and
reduction rules. Firstly, for the upper bounds {\bf UB1}--{\bf UB3}, it
is easy to see that {\bf UB2} and {\bf UB3} can be computed in time
linear to the number of
edges (\ie, $\bigo(m)$); note that, sorting vertices in {\bf UB3} can be
conducted in linear time by counting sort~\cite{Book:CSRL}.
For {\bf UB1}, we use the widely adopted greedy approach to assign
colors to vertices~\cite{KDD19:Chang,WALCOM10:Tomita}; that is, colors
are assigned to vertices in the reverse order of the degeneracy ordering
(see Definition~\ref{definition:degeneracy_ordering} for the definition
of degeneracy ordering), and a vertex
is assigned the smallest color that has not been taken by its neighbors.
Consequently, {\bf UB1} can be computed in $\bigo(m)$ time.

Secondly, for the reduction rules {\bf RR1}--{\bf RR5}, it is easy to
see that {\bf RR1},
{\bf RR2} and {\bf RR3} can be exhaustively applied until convergence
(\ie, until the instance can no longer be reduced by these reduction
rules) in linear time.
For {\bf RR4}, we do not apply it exhaustively for the
sake of efficiency. Instead, we let $u$ be the vertex most recently
added to $S$, and loop through each vertex $v \in V(g)\setminus S$ only
once. Observing that $cn(u,v) = |N_{\notS}(u) \cap N(v)|$, $cnon(u,v) =
|\overline{N}_{\notS}(u)| - |\overline{N}_{\notS}(u)\cap N(v)|$, and
$xn(u,v) = |V(g)\setminus S'| - cn(u,v) - cnon(u,v)$, applying {\bf RR4}
for $u$ and $v$ can be conducted in $\bigo(d_g(v))$ time by marking
$N_{\notS}(u)$ and $\overline{N}_{\notS}(u)$ in a preprocessing step.
Consequently, applying ${\bf RR4}$ once for all vertices 
$v \in V(g)\setminus S$ takes $\bigo(m)$ time in total.
For the reduction rule {\bf RR5}, it actually reduces the graph $g$ to
its $(\lb-k)$-core (see Definition~\ref{definition:k_core}), \ie, a
vertex is removed from $g$ if its
degree in $g$ is smaller than $\lb-k$; this can be conducted in $\bigo(m)$
time~\cite{Book18:Chang,JACM83:Matula}. Note that, if a vertex of $S$ is
removed during the process, then the instance $(g,S)$ is pruned (based
on {\bf UB2}).

Thirdly, for the reduction rule {\bf RR6}, we only apply it in the
preprocessing (\ie, Line~2 of \calg\ref{alg:kdc}) as it has a higher time
complexity than other reduction rules. Specifically, exhaustively
applying {\bf RR6} actually reduces
the input graph $G$ to its $(\lb-k+1)$-truss (see
Definition~\ref{definition:k_truss}), \ie, an edge $(u,v)$ is removed
from $G$ if the number of common neighbors of its two end-points in $G$
is smaller than $\lb-k-1$; this can be conducted in $\bigo(\delta(G)\times m)$
time~\cite{PVLDB12:Wang}, where $\delta(G)$ is the
degeneracy of $G$ and is at most $\sqrt{m}$.

In summary, we apply reduction rules {\bf RR1}--{\bf RR5} and upper
bounds {\bf UB1}--{\bf UB3} in \bbsearch;
thus, $P_1$ in \cthm\ref{theorem:time_complexity} is
$\bigo(m)$. For Line~2 of \calg\ref{alg:kdc}, we first exhaustively
apply {\bf RR5}, and then {\bf RR6}.
Thus, Line~2 of \calg\ref{alg:kdc} takes $\bigo(\delta(G)\times m)$
time.

\subsection{Compute a Large Initial Solution}
\label{sec:initial}

In this subsection, we discuss how to efficiently compute a large
initial $k$-defective clique at Line~1 of \calg\ref{alg:kdc}.
Firstly, we can heuristically compute a $k$-defective clique in
$\bigo(m)$ time based on the degeneracy ordering, \ie, the longest
suffix of the degeneracy ordering that is a $k$-defective clique.
The pseudocode is shown in \calg\ref{alg:degen}, denoted \kw{Degen}.

\begin{algorithm}[htb]
\small
\caption{$\kw{Degen}(G,k)$}
\label{alg:degen}
\KwIn{A graph $G$ and an integer $k$}
\KwOut{A large $k$-defective clique in $G$}

\vspace{1pt}
\State{Compute a degeneracy ordering for the vertices of $G$}
\State{$C \la$ the longest suffix of the degeneracy ordering that is a
$k$-defective clique}
\Return{$C$}
\end{algorithm}

As discussed at the end of \csec\ref{sec:time_complexity}, Line~2 of
\calg\ref{alg:kdc} takes $\bigo(\delta(G)\times m)$ time; this is
higher than the time complexity of \kw{Degen}. Thus, it makes
sense to spend a little more time at Line~1 of \calg\ref{alg:kdc} aiming
to compute a larger initial solution. Motivated by this, besides
heuristically computing a degeneracy ordering-based solution in the
input graph $G$, we also extract $n$ subgraphs from $G$ --- one
subgraph for each vertex of $G$ --- and heuristically compute a
degeneracy ordering-based solution in each of the subgraphs; the largest
one among these $n+1$ solutions is then kept as the initial solution.
To bound the time complexity by $\bigo(\delta(G)\times m)$, we extract
the subgraphs based on a degeneracy ordering of $G$. Specifically, let
$(v_1,v_2,\ldots,v_n)$ be a degeneracy ordering of $G$, the subgraph
extracted for $v_i$ then is the subgraph of $G$ induced by the set of higher
ranked neighbors of $u$ regarding the degeneracy ordering, \ie,
$N(v_i)\cap \{v_{i+1},\ldots,v_n\}$. The pseudocode is shown in
\calg\ref{alg:degen_opt}, denoted \kw{Degen\text{-}opt}. As the subgraph
extracted for $v_i$ will have at most $\min(d(v_i), \delta(G))$
vertices, the time complexity of \calg\ref{alg:degen_opt} is bounded by
\[\textstyle
	\sum_{i=1}^n \min(d(v_i), \delta(G))^2 \leq \sum_{i=1}^n d(v_i)
	\times \delta(G) = 2\times \delta(G)\times m
\]
Consequently, by invoking \calg\ref{alg:degen_opt} at Line~1 of
\calg\ref{alg:kdc}, $P_2$ in \csec\ref{sec:practical_kdc} is
$\bigo(\delta(G)\times m)$, and the time complexity of \kdc is
$\bigo(\delta(G)\times m + m\times \gamma_k^n)$.

\begin{algorithm}[htb]
\small
\caption{$\kw{Degen\text{-}opt}(G,k)$}
\label{alg:degen_opt}
\KwIn{A graph $G$ and an integer $k$}
\KwOut{A large $k$-defective clique in $G$}

\vspace{1pt}
\State{$C \la \kw{Degen}(G,k)$}
\State{Compute a degeneracy ordering for the vertices of $G$}
\ForEach{vertex $u \in V(G)$}{
	\State{$N^+(u) \la$ the set of higher ranked neighbors of $u$ in
	$G$, according to the degeneracy ordering}
	\State{$g \la $ the subgraph of $G$ induced by $N^+(u)$}
	\State{$C' \la \kw{Degen}(g, k)$}
	 \lIf{$|C'\cup u| > |C|$}{$C \la C'\cup u$}
}
\Return{$C$}
\end{algorithm}

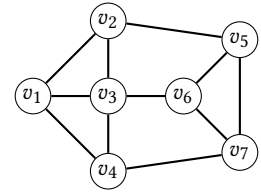
\begin{figure}[ht]
\centering
\begin{tikzpicture}[scale=0.5]
\node[draw, circle, inner sep=1.5pt] (1) at (0,0) {$v_1$};
\node[draw, circle, inner sep=1.5pt] (2) at (2,2) {$v_2$};
\node[draw, circle, inner sep=1.5pt] (3) at (2,0) {$v_3$};
\node[draw, circle, inner sep=1.5pt] (4) at (2,-2) {$v_4$};
\node[draw, circle, inner sep=1.5pt] (5) at (5.5,1.5) {$v_5$};
\node[draw, circle, inner sep=1.5pt] (6) at (4,0) {$v_6$};
\node[draw, circle, inner sep=1.5pt] (7) at (5.5,-1.5) {$v_7$};
\path[draw,thick] (1) -- (2) -- (3) -- (4) -- (1) -- (3) -- (6) -- (7)
-- (4);
\path[draw,thick] (2) -- (5) -- (7);
\path[draw,thick] (5) -- (6);
\end{tikzpicture}
\caption{Running example for computing an initial solution}
\label{fig:init}
\end{figure}

\begin{example}
Consider the graph in \cfig\ref{fig:init}, a degeneracy ordering is
$(v_1,v_2,v_5,v_3,v_4,v_6,v_7)$. Suppose $k=1$, the longest suffix that
is a $k$-defective clique is $\{v_4,v_6,v_7\}$ of size $3$; thus,
\kw{Degen} finds an initial solution of size $3$. Now, let's consider
$v_1$, its set of higher ranked neighbors is $N^+(v_1) =
\{v_2,v_3,v_4\}$. It is easy to see that the subgraph induced by
$N^+(v_1)$ is a $k$-defective clique. Thus, \kw{Degen\text{-}opt}
reports the initial solution $\{v_1,v_2,v_3,v_4\}$ of size $4$.
\end{example}

\section{Experiments}
\label{sec:experiment}

We have shown in \csec\ref{sec:framework} that our algorithm \kdc
achieves a better time complexity than the existing algorithms for the
problem of maximum $k$-defective clique computation.
In this section, we show empirically that \kdc also performs better than
the existing algorithms in practice. Specifically, we evaluate \kdc
against the following existing algorithms.
\begin{itemize}
	\item \kdbb: the existing algorithm with the state-of-the-art
		practical performance proposed in~\cite{AAAI22:Gao}.
	\item \madecp: the existing algorithm with the state-of-the-art time
		complexity proposed in~\cite{COR21:Chen}.
\end{itemize}
In addition, we also evaluate the following variants of our algorithm
\kdc to test the effectiveness of the different components of \kdc.
\begin{itemize}
	\item \kdcwob: \kdc without the upper bound {\bf UB1}.
	\item \kdcwor: \kdc without the reduction rules {\bf RR3} and {\bf
		RR4}.
	\item \kdcwoi: \kdc with the initial solution (\ie, Line~1 of
		\calg\ref{alg:kdc}) computed by \kw{Degen} and without applying the
		reduction rule {\bf RR6} at Line~2 of \calg\ref{alg:kdc}.
\end{itemize}
All our algorithms are implemented in C++ and compiled with -O3
optimization.~\footnote{The source code of \kdc is released at
\url{https://lijunchang.github.io/Maximum-kDC/}}
All experiments are run in the single-thread mode on a machine with an
Intel Core i7-8700 CPU and 64GB main memory and running Ubuntu 18.04.

\stitle{Datasets.} Same as~\cite{AAAI22:Gao}, we run
the algorithms on the following \highlight{three} graph collections.
\begin{itemize}
	\item The {\bf real-world graphs}
		collection~\footnote{\url{http://lcs.ios.ac.cn/~caisw/Resource/realworld\%20graphs.tar.gz}}
		contains $139$ real-world graphs from the Network Data Repository
		with up to $5.87\times 10^7$ vertices and $1.06\times 10^8$
		undirected edges.
	\item The {\bf Facebook graphs} collection~\footnote{\url{https://networkrepository.com/socfb.php}}
		contains $114$ Facebook social networks from the Network Data
		Repository with up to $5.92\times 10^7$ vertices and $9.25\times
		10^7$ undirected edges.
	\item \highlight{The {\bf DIMACS10\&SNAP graphs} collection contains
			37 graphs with
			up to $1.04\times 10^6$ vertices and $6.89\times 10^6$
			undirected edges. Among the $37$ graphs, 
			$27$ are from 
			DIMACS10~\footnote{\url{https://www.cc.gatech.edu/dimacs10/downloads.shtml}}
			and $10$ are from
			SNAP~\footnote{\url{http://snap.stanford.edu/data/}}.
			}
\end{itemize}
\highlight{Note that, the graphs included in these three collections are
the same ones tested in~\cite{AAAI22:Gao}.}

\stitle{Metric.} 
We record the total {\em processing time} 
of running an algorithm on a graph instance for a specific $k$.
The recorded processing time is the total CPU
time excluding the I/O time of loading the graph instance from disk
to main memory.
%
%
\highlight{Same as~\cite{AAAI22:Gao}, we choose $k$ from
	$\{1,3,5,10,15,20\}$ and set a time limit of $3$ hours for each
testing.}

\begin{table}[t]
\footnotesize
\setlength{\tabcolsep}{2pt}
\centering
\caption{Number of solved instances by the algorithms \kdc, \kdbb and
	\madecp with a time limit of $3$ hours (best
performers are highlighted in bold)}
\label{table:against_existing}
\begin{tabular}{c|ccc|ccc|ccc}
\hline
& \multicolumn{3}{c|}{Real-world graphs} & \multicolumn{3}{c|}{Facebook
graphs} & \multicolumn{3}{c}{\highlight{DIMACS10\&SNAP}} \\ 
& \kdc & \kdbb & \madecp & \kdc & \kdbb & \madecp & \kdc & \kdbb &
\madecp \\ \hline
$k = 1$ & {\bf 133} & 117 & 115 & {\bf 114} & 110 & 110 & {\bf 37} & 36 & 36 \\
$k=3$ & {\bf 130} & 107 & 94 & {\bf 114} & 110 & 104 & {\bf 37} & 35 & 31 \\
$k=5$ & {\bf 127} & 104 & 81 & {\bf 114} & 108 & 78 & {\bf 37} & 34 & 28 \\
$k=10$ & {\bf 119} & 85 & 36 & {\bf 111} & 109 & 9 & {\bf 36} & 30 & 15 \\
$k=15$ & {\bf 110} & 68 & 26 & 101 & {\bf 103} & 0 & {\bf 29} & 25 & 10 \\
$k=20$ & {\bf 104} & 56 & 20 & {\bf 88} & 80 & 0 & {\bf 27} & 22 & 6 \\ \hline
\end{tabular}
\end{table}

\subsection{Against the Existing Algorithms}

In this subsection, we evaluate our algorithm \kdc against the 
existing algorithms \kdbb and \madecp, regarding the efficiency. The
results on the number of
solved instances with a time limit of 3 hours are shown in
\ctab\ref{table:against_existing}, which are also partially illustrated
in Figures~\ref{fig:realworld} and \ref{fig:facebook}.
\madecp is an improved version, by the authors of
\kdbb~\cite{AAAI22:Gao}, of the \madec algorithm proposed
in~\cite{COR21:Chen}.
As {\em the authors of \cite{AAAI22:Gao} are not able to provide the
code of algorithms \kdbb and \madecp}, their numbers reported
in~\ctab\ref{table:against_existing} are obtained from the original
paper~\cite{AAAI22:Gao}; \highlight{note that~\cite{AAAI22:Gao} tests
exactly the same sets of graphs for the three graph collections and also has
the time limit of $3$ hours}.
From \ctab\ref{table:against_existing}, we can see that \kdbb
significantly outperforms \madecp (especially for $k \geq 5$), and our
algorithm \kdc further outperforms \kdbb with the only exception of
$k=15$ on the Facebook graphs collection.
We also would like to highlight two other observations that can be
observed from Figures~\ref{fig:realworld} and \ref{fig:facebook}.
Firstly, on the real-world graphs collection, \kdc with a time limit of
$3$ {\em seconds} solves even more instances than \kdbb with a time limit
of $3$ {\em hours}.
Secondly, on the Facebook graphs collection, \kdc solves all $114$
instances with time limits of $125$, $393$ and $1353$ seconds,
respectively, for $k=1$, $3$ and $5$. This demonstrates the practical
superiority of \kdc over the existing algorithms.

{
\addtolength{\tabcolsep}{-2pt}
\begin{table*}[ht]
\centering
\footnotesize
\caption{\highlight{Processing time (in seconds) of \kdc, \kdcwor,
		\kdcwob, \kdcwoi and \kdbb on the $41$ Facebook
		graphs with more than $15,000$ vertices; the results of \kdcwor,
		\kdcwob and \kdcwoi for $k=1$ and $k=5$ are omitted due to
		space limitations. Best performers are highlighted
in bold; if a running time is slow than the fastest running time by less
than 10\%, we also consider it to be best. $n$ is the number of vertices
and $m$ is the number of edges in the graph.}}
\label{table:facebook_time}
\begin{tabular}{crr|cc|ccccc|cc|ccccc}
\hline
 & & & \multicolumn{2}{c|}{$k=1$} & \multicolumn{5}{c|}{$k=3$} & \multicolumn{2}{c|}{$k=5$} & \multicolumn{5}{c}{$k=10$} \\ 
 & $n$ & $m$ & \kdc & \kdbb & \kdc & \kdcwor & \kdcwob & \kdcwoi & \kdbb & \kdc & \kdbb & \kdc & \kdcwor & \kdcwob & \kdcwoi & \kdbb \\ \hline
A-anon & 3M & 23M & {\bf 5.0} & - & {\bf 5.2} & {\bf 5.1} & 29 & 6803 & - & {\bf 5.7} & - & {\bf 73} & 412 & 6540 & - & -\\
Auburn71 & 18K & 973K & {\bf 1.4} & 432 & 1.9 & {\bf 1.6} & 384 & 55 & 536 & {\bf 8.6} & 639 & {\bf 956} & {\bf 905} & - & - & 1195\\
B-anon & 2M & 20M & {\bf 7.9} & - & {\bf 8.4} & {\bf 8.1} & 57 & - & - & {\bf 9.2} & - & {\bf 44} & 56 & 7858 & - & -\\
Berkeley13 & 22K & 852K & {\bf 0.18} & 425 & {\bf 0.18} & {\bf 0.18} & 0.42 & 6.5 & 452 & {\bf 0.19} & 506 & {\bf 0.34} & 0.39 & 61 & 55 & 630\\
BU10 & 19K & 637K & {\bf 0.09} & 252 & {\bf 0.15} & {\bf 0.15} & 1.1 & 6.5 & 290 & {\bf 0.39} & 332 & {\bf 4.0} & 5.1 & 16 & 35 & 370\\
Cornell5 & 18K & 790K & {\bf 1.1} & 393 & {\bf 2.1} & {\bf 2.0} & 194 & 249 & 922 & {\bf 2.6} & 1265 & {\bf 17} & {\bf 16} & 8670 & - & 2636\\
FSU53 & 27K & 1M & {\bf 0.35} & 209 & 0.40 & {\bf 0.23} & - & 80 & 610 & {\bf 2.8} & 828 & 248 & {\bf 221} & - & 8351 & 1400\\
Harvard1 & 15K & 824K & {\bf 0.76} & 347 & {\bf 0.82} & {\bf 0.85} & 91 & 284 & 421 & {\bf 0.94} & 517 & {\bf 11} & {\bf 10} & 3680 & - & 1354\\
Indiana & 29K & 1M & {\bf 0.46} & 1142 & {\bf 0.46} & {\bf 0.48} & 21 & 95 & 1138 & {\bf 0.53} & 1261 & {\bf 19} & {\bf 20} & 1975 & 3710 & 1421\\
Indiana69 & 29K & 1M & {\bf 0.46} & 1134 & {\bf 0.46} & {\bf 0.48} & 21 & 97 & 1072 & {\bf 0.54} & 1186 & {\bf 19} & {\bf 20} & 1964 & 3706 & 1321\\
konect & 59M & 92M & {\bf 8.1} & - & {\bf 8.2} & {\bf 7.9} & 9.7 & 178 & - & {\bf 9.4} & - & - & - & - & - & -\\
Maryland58 & 20K & 744K & {\bf 0.10} & 150 & {\bf 0.12} & {\bf 0.11} & 0.18 & 4.1 & 162 & {\bf 0.12} & 185 & {\bf 0.60} & 0.82 & 2.7 & 8.2 & 239\\
Michigan23 & 30K & 1M & {\bf 0.63} & 833 & {\bf 0.66} & {\bf 0.63} & 0.88 & 1556 & 1072 & {\bf 0.67} & 971 & {\bf 2.2} & 2.9 & 215 & - & 1384\\
MSU24 & 32K & 1M & {\bf 0.35} & 493 & {\bf 0.35} & {\bf 0.33} & 0.40 & 92 & 576 & {\bf 0.34} & 666 & {\bf 0.47} & {\bf 0.50} & 1.5 & 10227 & 879\\
MU78 & 15K & 649K & {\bf 0.10} & 182 & {\bf 0.13} & {\bf 0.12} & 0.53 & 1.1 & 200 & {\bf 0.32} & 215 & {\bf 67} & {\bf 68} & 393 & 203 & 306\\
NYU9 & 21K & 715K & {\bf 0.09} & 349 & {\bf 0.09} & {\bf 0.09} & 0.13 & 17 & 399 & {\bf 0.09} & 396 & {\bf 0.12} & 0.13 & 0.17 & 26 & 466\\
Oklahoma97 & 17K & 892K & {\bf 0.78} & 383 & {\bf 0.96} & {\bf 0.89} & - & 1162 & 2048 & {\bf 5.1} & 3938 & 379 & {\bf 334} & - & 10533 & 6926\\
OR & 63K & 816K & {\bf 0.16} & 356 & {\bf 0.30} & 0.34 & 21 & 9.9 & 456 & {\bf 1.0} & 587 & {\bf 55} & {\bf 55} & 885 & 258 & 1486\\
Penn94 & 41K & 1M & {\bf 0.23} & 1139 & {\bf 0.23} & {\bf 0.22} & 0.25 & 8.2 & 1557 & {\bf 0.23} & 1820 & {\bf 0.29} & {\bf 0.32} & 0.35 & 20 & 1972\\
Rutgers89 & 24K & 784K & {\bf 0.08} & 219 & 0.07 & 0.07 & 0.09 & {\bf 0.04} & 276 & {\bf 0.08} & 279 & {\bf 0.20} & {\bf 0.22} & 1.4 & 7.1 & 386\\
Tennessee95 & 16K & 770K & {\bf 0.54} & 246 & {\bf 0.56} & {\bf 0.53} & 2.7 & 16 & 361 & {\bf 0.52} & 424 & {\bf 1.8} & {\bf 1.7} & 62 & 884 & 554\\
Texas80 & 31K & 1M & {\bf 0.56} & 342 & {\bf 0.66} & {\bf 0.62} & 4.6 & 52 & 423 & {\bf 0.73} & 534 & 80 & {\bf 70} & 1136 & 2603 & 753\\
Texas84 & 36K & 1M & {\bf 6.2} & 1490 & 13 & {\bf 11} & 6503 & 5555 & 1674 & {\bf 69} & 2769 & 1321 & {\bf 1134} & - & - & 10253\\
UC33 & 16K & 522K & {\bf 0.07} & 156 & {\bf 0.07} & {\bf 0.06} & 1.6 & 1.3 & 171 & {\bf 0.07} & 181 & {\bf 0.14} & {\bf 0.15} & 148 & 3.2 & 263\\
uci-uni & 58M & 92M & {\bf 13} & - & {\bf 13} & {\bf 12} & 14 & 206 & - & {\bf 14} & - & - & - & - & - & -\\
UCLA & 20K & 747K & {\bf 0.09} & 190 & 0.09 & 0.09 & 0.11 & {\bf 0.04} & 206 & {\bf 0.09} & 237 & {\bf 0.14} & {\bf 0.15} & 0.17 & 4.9 & 290\\
UCLA26 & 20K & 747K & {\bf 0.09} & 184 & 0.09 & 0.09 & 0.12 & {\bf 0.04} & 207 & {\bf 0.09} & 215 & {\bf 0.14} & {\bf 0.15} & 0.19 & 4.9 & 288\\
UConn & 17K & 604K & {\bf 0.06} & 109 & 0.06 & 0.05 & 0.06 & {\bf 0.04} & 126 & {\bf 0.06} & 169 & {\bf 0.13} & 0.16 & 0.22 & 2.5 & 194\\
UConn91 & 17K & 604K & {\bf 0.06} & 105 & 0.06 & 0.05 & 0.07 & {\bf 0.04} & 123 & {\bf 0.06} & 173 & {\bf 0.13} & 0.16 & 0.22 & 2.5 & 208\\
UF & 35K & 1M & {\bf 0.58} & 793 & {\bf 0.58} & 0.84 & - & 282 & 1332 & {\bf 0.74} & 1602 & {\bf 27} & {\bf 29} & - & 8777 & 2579\\
UF21 & 35K & 1M & {\bf 0.57} & 787 & {\bf 0.58} & 0.84 & - & 281 & 1297 & {\bf 0.74} & 1542 & {\bf 27} & {\bf 29} & - & 8767 & 2571\\
UGA50 & 24K & 1M & {\bf 5.7} & 724 & 43 & {\bf 37} & - & 4895 & 1467 & {\bf 165} & 2459 & 3318 & {\bf 2856} & - & - & 6794\\
UIllinois & 30K & 1M & {\bf 0.68} & 486 & {\bf 0.69} & {\bf 0.65} & 2.9 & 93 & 644 & {\bf 0.65} & 806 & {\bf 3.6} & {\bf 3.5} & 342 & 8237 & 1245\\
UIllinois20 & 30K & 1M & {\bf 0.68} & 486 & {\bf 0.68} & {\bf 0.65} & 2.9 & 93 & 610 & {\bf 0.66} & 784 & {\bf 3.6} & {\bf 3.5} & 341 & 8195 & 1217\\
UMass92 & 16K & 519K & {\bf 0.15} & 226 & {\bf 0.15} & {\bf 0.15} & 0.19 & 25 & 245 & {\bf 0.17} & 265 & {\bf 0.30} & 0.37 & 0.58 & 82 & 318\\
UNC28 & 18K & 766K & {\bf 0.46} & 236 & {\bf 0.47} & {\bf 0.45} & 0.82 & 54 & 287 & {\bf 0.46} & 336 & {\bf 2.1} & {\bf 2.1} & 15 & 7278 & 380\\
USC35 & 17K & 801K & {\bf 0.31} & 232 & {\bf 0.31} & {\bf 0.30} & 0.60 & 390 & 267 & {\bf 0.31} & 334 & 0.52 & {\bf 0.47} & 7.7 & 6226 & 409\\
UVA16 & 17K & 789K & {\bf 0.43} & 341 & {\bf 0.45} & {\bf 0.49} & 2.4 & 130 & 387 & {\bf 0.57} & 400 & {\bf 14} & 19 & 310 & 8666 & 552\\
Virginia63 & 21K & 698K & {\bf 0.29} & 84 & {\bf 0.29} & {\bf 0.29} & 0.34 & 1.3 & 103 & {\bf 0.26} & 143 & {\bf 1.1} & {\bf 1.1} & 2.9 & 169 & 215\\
Wisconsin87 & 23K & 835K & {\bf 0.17} & 532 & {\bf 0.18} & {\bf 0.18} & 9.7 & 19 & 612 & {\bf 0.31} & 664 & {\bf 47} & {\bf 43} & 1323 & 292 & 924\\
wosn-friends & 63K & 817K & {\bf 0.16} & 375 & {\bf 0.30} & 0.34 & 21 & 10.0 & 438 & {\bf 1.0} & 533 & {\bf 54} & {\bf 55} & 895 & 259 & 1260\\
\hline
\end{tabular}

\end{table*}
\addtolength{\tabcolsep}{2pt}
}

\begin{figure*}[t]
\centering
\begin{minipage}{.48\textwidth}
\begin{center}
\includegraphics[scale=.4]{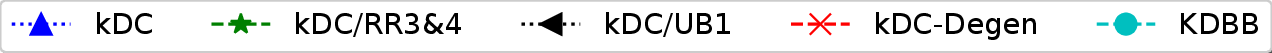}
\end{center}
\vspace*{-12pt}
\subfigure[$k=1$]{
	\includegraphics[width=.5\columnwidth]{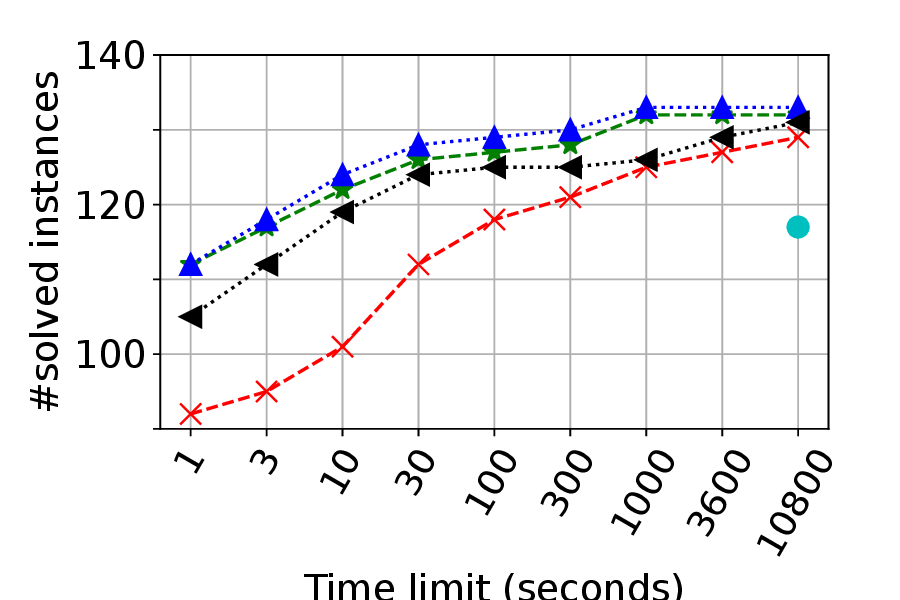}
}%
\subfigure[$k=3$]{
	\includegraphics[width=.5\columnwidth]{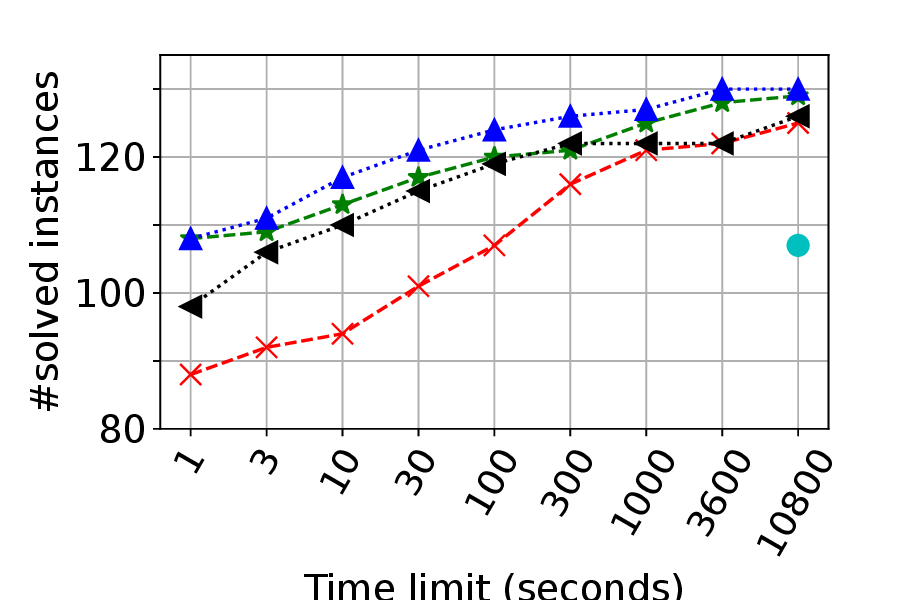}
}
\subfigure[$k=5$]{
	\includegraphics[width=.5\columnwidth]{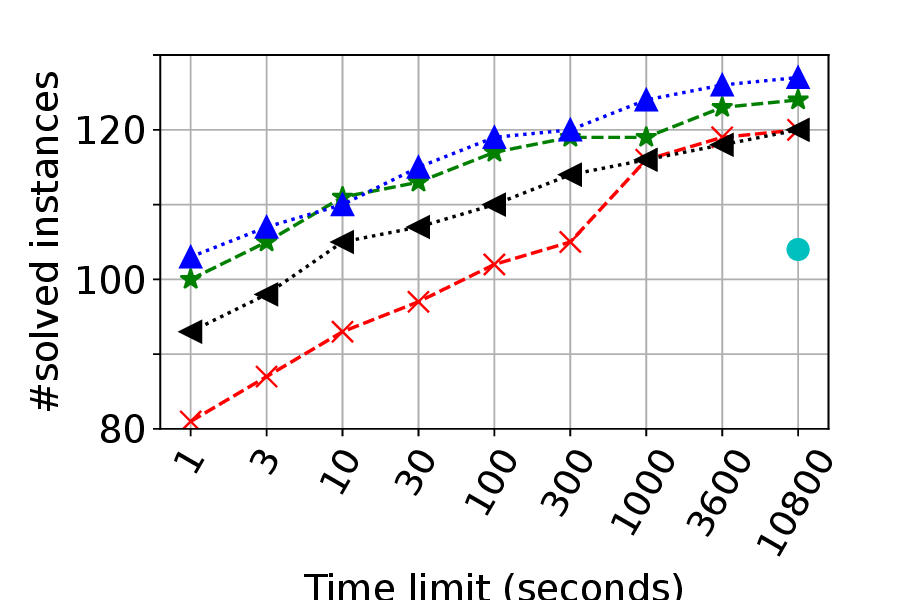}
}%
\subfigure[$k=10$]{
	\includegraphics[width=.5\columnwidth]{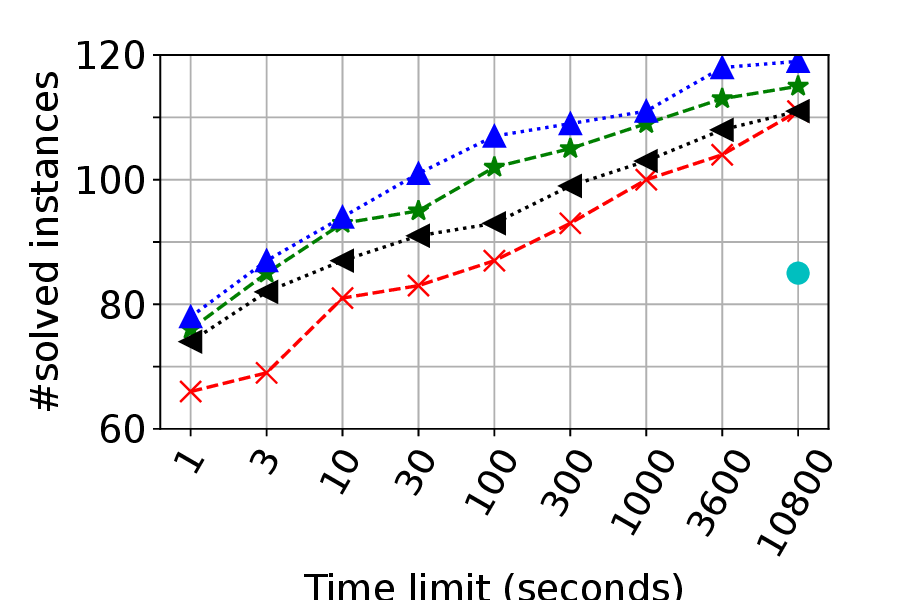}
}
\subfigure[$k=15$]{
	\includegraphics[width=.5\columnwidth]{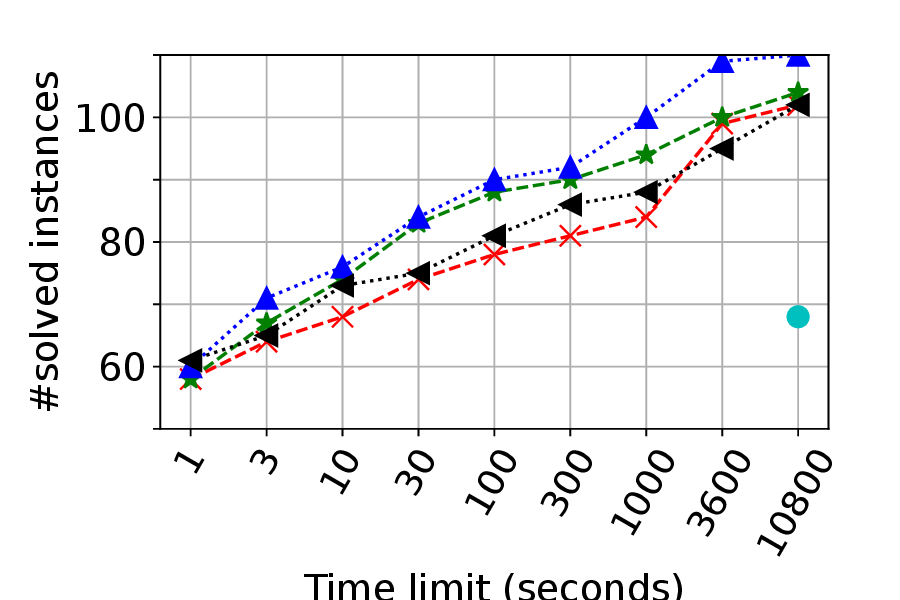}
}%
\subfigure[$k=20$]{
	\includegraphics[width=.5\columnwidth]{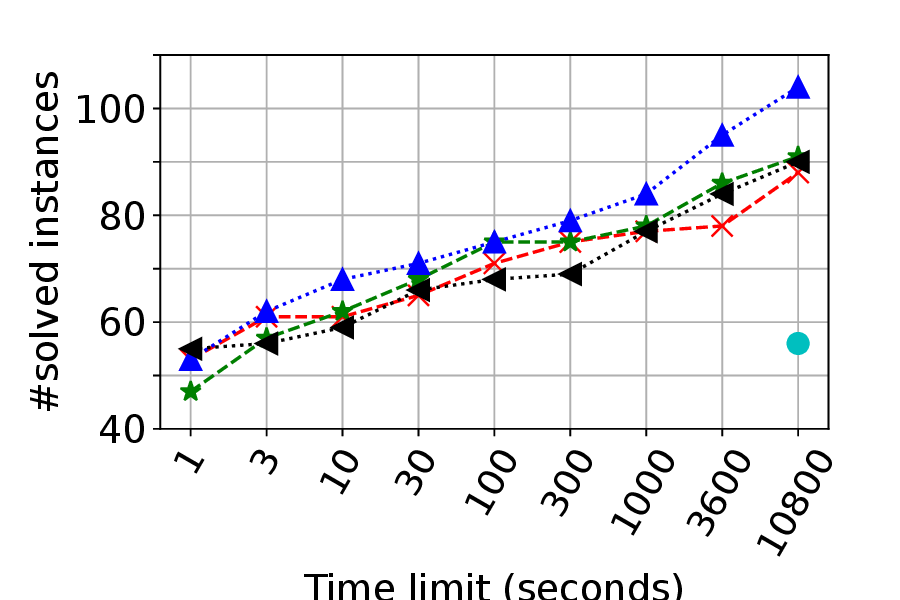}
}
\vspace{-2pt}
\caption{Number of solved instances for real-world graphs (vary time
limit, best viewed in color)}
\label{fig:realworld}
\end{minipage}
\hspace*{5pt}
\begin{minipage}{.48\textwidth}
\begin{center}
\includegraphics[scale=.4]{exp/title}
\end{center}
\vspace*{-12pt}
\subfigure[$k=1$]{
	\includegraphics[width=.5\columnwidth]{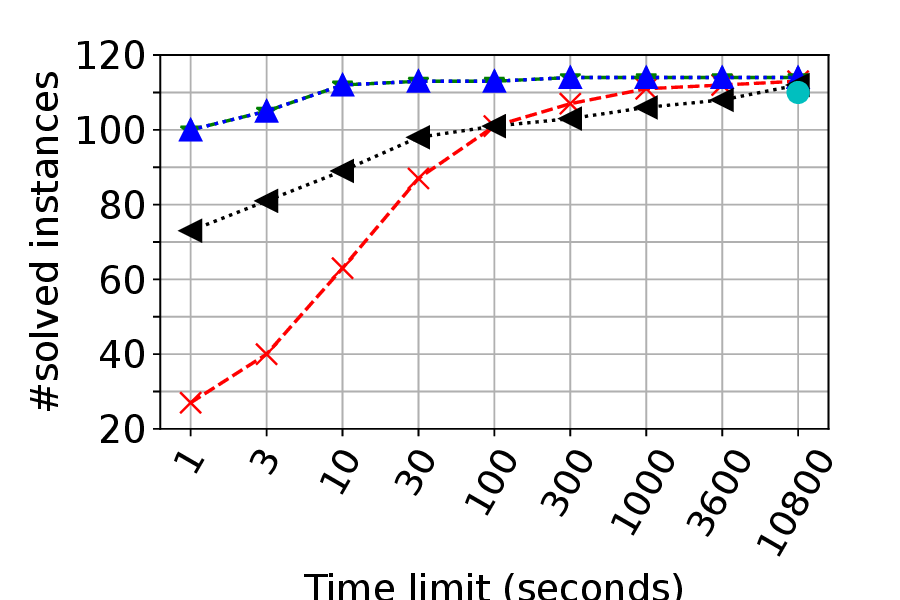}
}%
\subfigure[$k=3$]{
	\includegraphics[width=.5\columnwidth]{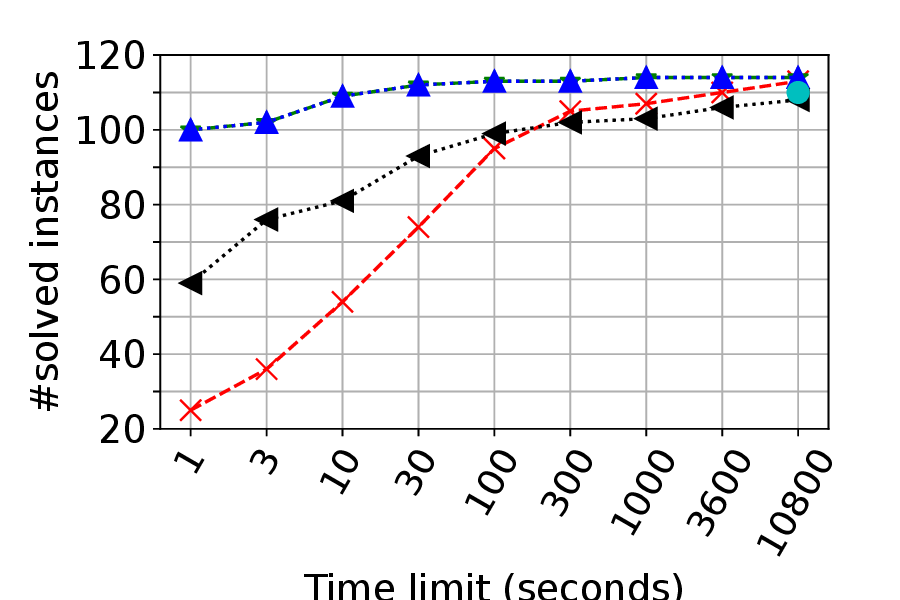}
}
\subfigure[$k=5$]{
	\includegraphics[width=.5\columnwidth]{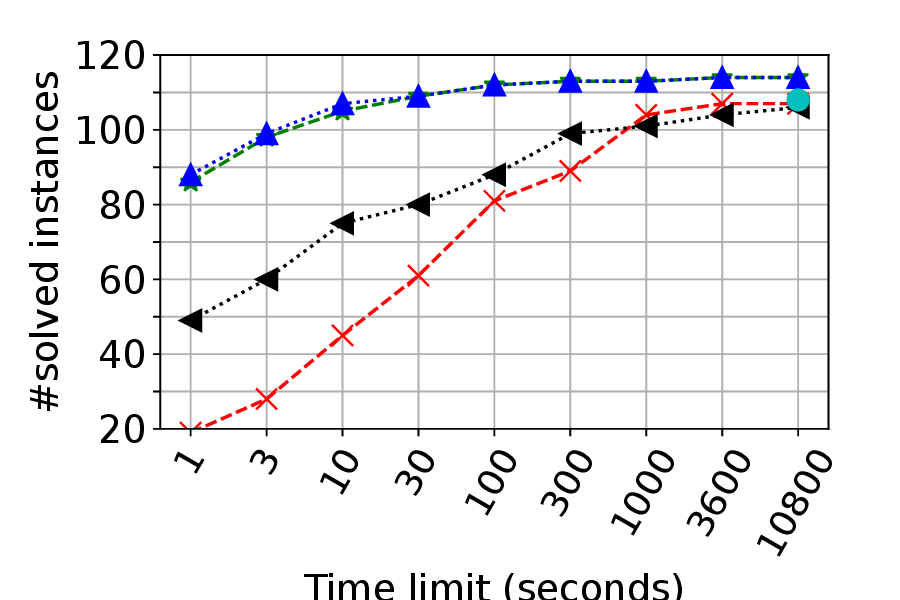}
}%
\subfigure[$k=10$]{
	\includegraphics[width=.5\columnwidth]{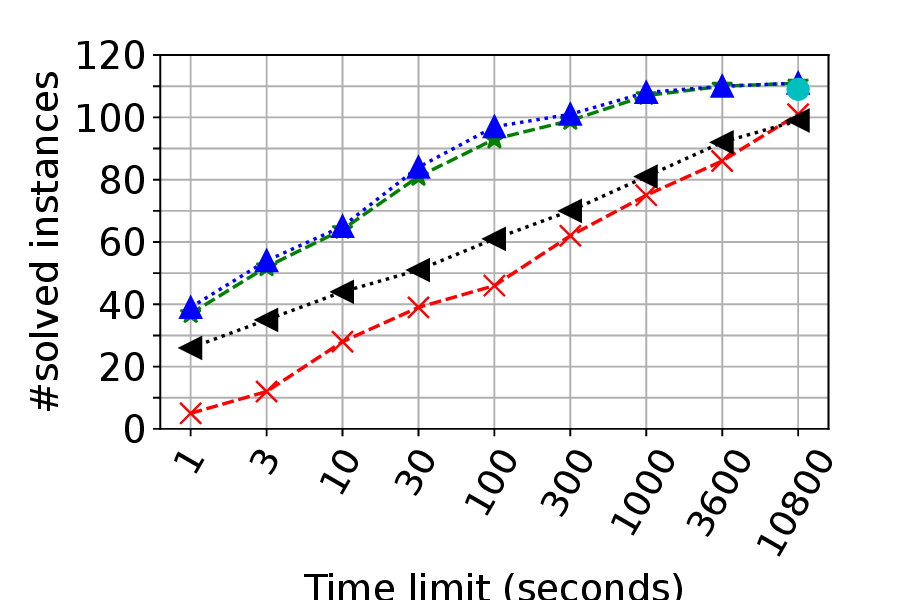}
}
\subfigure[$k=15$]{
	\includegraphics[width=.5\columnwidth]{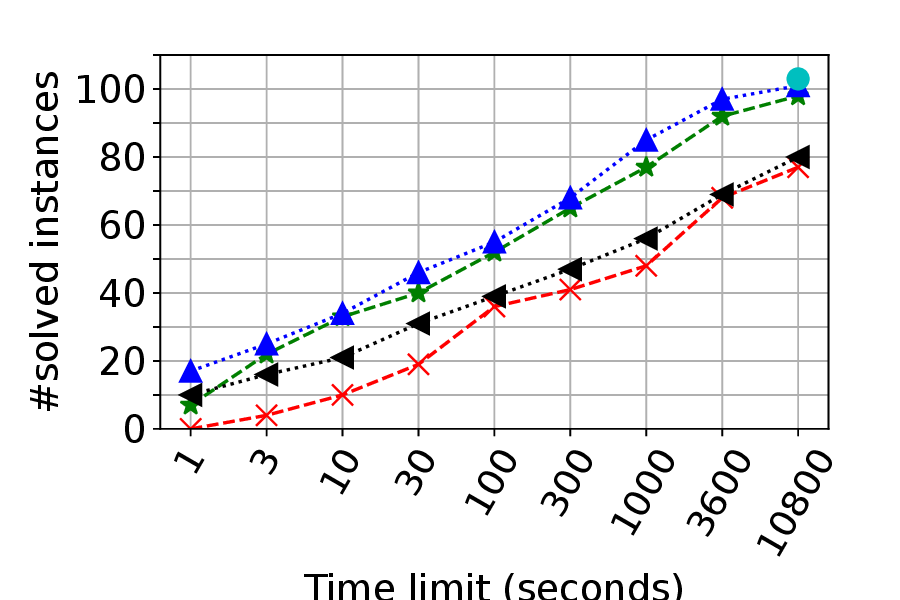}
}%
\subfigure[$k=20$]{
	\includegraphics[width=.5\columnwidth]{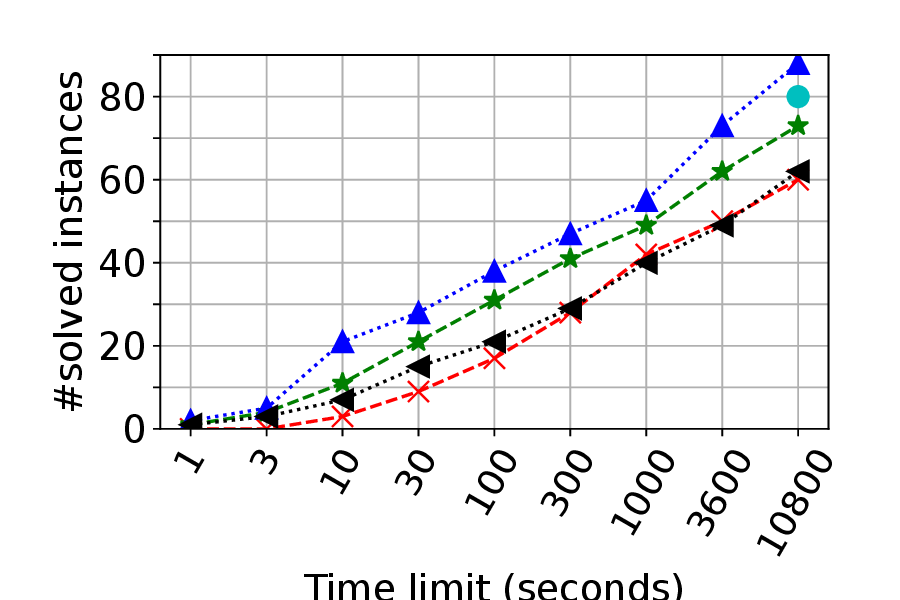}
}
\vspace{-2pt}
\caption{Number of solved instances for Facebook graphs (vary time
limit, best viewed in color)}
\label{fig:facebook}
\end{minipage}
\end{figure*}

To dive into a more detailed performance gain of \kdc over \kdbb, we report
the actual processing time of \kdc and \kdbb on the subset of Facebook
graphs that have more than $15,000$ vertices; there are $41$ such
graphs. We report the results for $k=1$, $3$, $5$, and $10$ in
\ctab\ref{table:facebook_time}, as \kdbb only gave results
for such $k$ values in~\cite{AAAI22:Gao};
\highlight{
for now, please ignore the columns regarding algorithms \kdcwor,
\kdcwob, and \kdcwoi.
The number of vertices and edges in these graphs are also illustrated in
\ctab\ref{table:facebook_time}.
Note that, an `$-$' for \kdbb indicates that the result is not available
(specifically, \cite{AAAI22:Gao} didn't report the results on the four
	graphs \texttt{A-anon}, \texttt{B-anon}, \texttt{konect},
\texttt{uci-uni}), while an `$-$' for our algorithms indicates that the
processing time is longer than the $3$-hour limit.
}
From \ctab\ref{table:facebook_time}, we can observe that \kdc
consistently and significantly runs faster than
\kdbb across all these testings.
In particular, \kdc on average is $1552$, $1754$, $1636$ and $820$ times
faster than \kdbb for $k=1$, $3$, $5$ and $10$, respectively.
This further demonstrates the superior performance of \kdc over
the existing fastest algorithm \kdbb.

\highlight{
We also would like to discuss the performance of our algorithm \kdc on
some large-scale graphs. Firstly, for the two graphs, konect and
uci-uni, in the Facebook graphs collection that have $58$M vertices and
$92$M undirected edges (here M means $\times 10^6$), \kdc is able to find the
maximum $k$-defective clique within the time limit for $k\leq 5$ and
times out for larger $k$ values; please refer to
\ctab\ref{table:facebook_time} for the results.
	Secondly, for the soc-orkut graph in
real-world graphs collection that have $3$M vertices and $106$M
undirected edges,
\kdc is able to find the maximum $k$-defective clique within the time
limit for $k\leq 3$ and times out for larger $k$ values. Thirdly, we
also tested \kdc on the webbase-2001 graph (downloaded from
\url{https://law.di.unimi.it/datasets.php}) that has $116$M vertices and
$855$M undirected edges. \kdc is able to find the maximum $k$-defective
clique for all the tested $k$ values within $30$ seconds. It will be our
future work to further improve the number of solved instances for the
different $k$ values.
}

\subsection{Ablation Studies}

Now, we conduct ablation studies for our proposed
techniques. Firstly, we compare \kdc with \kdcwor which is \kdc without
applying our two reduction rules {\bf RR3} and {\bf RR4} that are
described in \csec\ref{sec:reduction_rules}. The results for $k=1$, $3$,
$5$, $10$, $15$, and $20$ on the real-world
graphs collection are shown in \cfig\ref{fig:realworld}, and that on the
Facebook graphs collection are shown in \cfig\ref{fig:facebook};
specifically, we vary the time limit and report the number of graph
instances that are solved by an algorithm within a specific time limit.
We can see that \kdc consistently outperforms \kdcwor, and the
improvement is more evident when $k$ becomes large; for example,
for $k=20$ and with a time limit of $3$ hours, \kdc solves $13$ and
$15$ more instances than \kdbb on
the real-world graphs collection and the Facebook graphs collection,
respectively. This demonstrates that our new reduction rules {\bf RR3}
and {\bf RR4} are effective in improving the efficiency of maximum
$k$-defective computation.
\highlight{
However, we also observe that \kdc and \kdcwor perform similarly for $k
\leq 10$ on the Facebook graphs collection; specifically,
the processing time of \kdcwor on $41$ of the Facebook graphs for $k=3$
and $k=10$ are also listed in \ctab\ref{table:facebook_time}. One of
the reasons is that the upper bound {\bf UB1}, which is used in both
\kdc and \kdcwor, is very effective for these testings, and as a result
{\bf RR3} and {\bf RR4} do not prune many additional vertices. To verify
that, we also implement a version of \kdc without {\bf UB1}, {\bf RR3}
and {\bf RR4}, denoted \kdcwobr, and compare it with \kdcwob. Our
results show that \kdcwob solves $5$, $8$, and $14$ more instances than
\kdcwobr, respectively, for $k=3,5$ and $10$ on the Facebook graphs
collection with a time limit of $100$ seconds; this demonstrates that
{\bf RR3} and {\bf RR4} are effective in these settings when {\bf UB1}
is not applied.
}

Secondly, we evaluate \kdc against \kdcwob which is \kdc without
applying our upper bound {\bf UB1} as introduced in
\csec\ref{sec:upper_bound}. The
results are also reported in Figures~\ref{fig:realworld} and
\ref{fig:facebook} \highlight{and \ctab\ref{table:facebook_time}}. We can see that \kdc consistently outperforms
\kdcwob, and the improvement can be large, especially on the Facebook
graphs. This demonstrates that our upper bound {\bf UB1} is effective in
improving the performance of \kdc.
Also, we would like to remark that {\bf UB3} is the upper bound proposed
in \kdbb~\cite{AAAI22:Gao}, and is also used in \kdc. Thus, our upper
bound {\bf UB1} is
also tighter than the one proposed in~\cite{AAAI22:Gao}, as otherwise,
the performance of \kdc would be similar to that of \kdcwob.

\begin{table}[t]
\small
\centering
\newcommand{\kdcwo}{\kw{kDC\text{-}D}}
\caption{The difference of preprocessing results between \kdc and
	\kdcwoi ($C^0$
	denotes the initial solution obtained by Line~1 of
	\calg\ref{alg:kdc}; $n^0$ and $m^0$ respectively denote
the number of vertices and the number of edges in the reduced graph
obtained by Line~2 of \calg\ref{alg:kdc}; subscripts denote the
algorithms with \kdcwoi being abbreviated as \kdcwo)}
\label{table:preprocess}
\begin{tabular}{c|ccc|ccc}
\hline
& \multicolumn{3}{c|}{Real-world graphs} & \multicolumn{3}{c}{Facebook
graphs} \\ 
& $\frac{|C_{\kdc}^0|}{|C_{\kdcwo}^0|}$ &
$\frac{n_{\kdc}^0}{n_{\kdcwo}^0}$ & $\frac{m_\kdc^0}{m_\kdcwo^0}$ & $\frac{|C_{\kdc}^0|}{|C_{\kdcwo}^0|}$ &
$\frac{n_{\kdc}^0}{n_{\kdcwo}^0}$ & $\frac{m_\kdc^0}{m_\kdcwo^0}$ \\ \hline
$k=1$ & $1.19$ & $0.27$ & $0.26$ & $1.30$ & $0.03$ & $0.02$ \\
$k=3$ & $1.15$ & $0.47$ & $0.45$ & $1.26$ & $0.04$  & $0.03$ \\
$k=5$ & $1.13$ & $0.52$ & $0.52$ & $1.24$ & $0.06$ & $0.04$ \\
$k=10$ & $1.11$ & $0.63$ & $0.63$ & $1.21$ & $0.11$ & $0.08$ \\
$k=15$ & $1.09$ & $0.68$ & $0.69$ & $1.19$ & $0.16$ & $0.13$ \\
$k=20$ & $1.08$ & $0.73$ & $0.74$ & $1.18$ & $0.23$ & $0.19$ \\ \hline
\end{tabular}
\end{table}

Thirdly, we compare \kdc with \kdcwoi which is \kdc invoking \kw{Degen}
to compute the initial solution at Line~1 of \calg\ref{alg:kdc} and
without applying the reduction rule {\bf RR6} at Line~2 of
\calg\ref{alg:kdc}. As a result, the preprocessing of \kdcwoi (\ie, Lines~1--2
of \calg\ref{alg:kdc}) takes only $\bigo(m)$ time, in contrast to the
$\bigo(\delta(G)\times m)$ preprocessing time of \kdc. The experimental
results are again shown in Figures~\ref{fig:realworld} and
\ref{fig:facebook} \highlight{and \ctab\ref{table:facebook_time}}. We can see that \kdc consistently outperforms
\kdcwoi, and the gap is huge when both $k$ and the time limit are small,
\eg, when $k \leq 5$ and the time limit is at most $10$ seconds. To
explain this, we also show the difference of preprocessing results
between \kdc and \kdcwoi (\ie, Lines~1--2 of \calg\ref{alg:kdc}) in
\ctab\ref{table:preprocess}. We can see that \kdc computes a larger
initial solution and a smaller reduced graph than \kdcwoi; the
improvement is more significant when $k$ is small.

In summary, each of these additional techniques (\ie, 
reduction rules {\bf RR3} and {\bf RR4}, upper bound {\bf UB1}, 
and computing a large initial solution) improves the practical
efficiency of \kdc.

\subsection{Properties of Maximum $k$-Defective Clique}

\begin{table}[t]
\small
\setlength{\tabcolsep}{2pt}
\centering
\caption{The (average and maximum) ratio of the maximum $k$-defective
clique size over the maximum clique size \highlight{for each of the
three graph collections}}
\label{table:size_ratio}
\begin{tabular}{c|cc|cc|cc}
	\hline
	& \multicolumn{2}{c|}{Real-world graphs} &
	\multicolumn{2}{c|}{Facebook graphs} &
	\multicolumn{2}{c}{\highlight{DIMACS10\&SNAP}} \\ 
& Avg Ratio & Max Ratio & Avg Ratio & Max Ratio & Avg Ratio &  Max Ratio \\ \hline
$k=1$ & $1.067$ & $1.5$ & $1.032$ & $1.25$ & $1.046$ & $1.200$ \\
$k=3$ & $1.144$ & $2$ & $1.083$ & $1.5$ & $1.107$ & $1.400$ \\
$k=5$ & $1.201$ & $2$ & $1.118$ & $1.67$ & $1.169$ & $1.600$ \\
$k=10$ & $1.314$ & $2.5$ & $1.170$ & $1.75$ & $1.243$ & $1.800$ \\
$k=15$ & $1.422$ & $3$ & $1.223$ & $2$ & $1.313$ & $2.000$ \\
$k=20$ & $1.516$ & $3.5$ & $1.264$ & $2.25$ & $1.370$ & $2.200$ \\ \hline
\end{tabular}
\end{table}

In this subsection, we analyze the properties of maximum $k$-defective
clique.
%
Firstly, we compare the maximum $k$-defective clique size,
computed by \kdc, with the maximum clique size, computed by
\mcbrb~\footnote{\url{https://lijunchang.github.io/MC-BRB/}}~\cite{KDD19:Chang},
on the \highlight{three} graph collections. \highlight{For each $k$ and each of the
three graph collections}, the results on the average and maximum
ratio of
$\frac{\text{maximum $k$-defective clique size}}{\text{maximum clique
size}}$ over all graphs that \kdc finishes within $3$
hours are reported in \ctab\ref{table:size_ratio} \highlight{(the total number of such
graphs can be found in \ctab\ref{table:against_existing})}; we remark that \mcbrb
successfully finds the maximum clique for all the graphs within the time
limit. From \ctab\ref{table:size_ratio}, we can see that on the
real-world graphs collection, the maximum
$k$-defective clique size is on average $31\%$ (and maximum $150\%$)
larger than the maximum
clique size for $k=10$, and is on average $51\%$ (and maximum $250\%$)
larger for $k=20$. This
demonstrates that the relaxation of $k$-defective clique indeed enables
us to find larger near-cliques.

\begin{table}[t]
\small
\centering
\caption{\highlight{Number of graphs where the maximum $k$-defective
		clique is an extension of a maximum clique}}
\label{table:extension_of_clique}
\begin{tabular}{c|ccc}
\hline
& Real-world & Facebook & DIMACS10\&SNAP \\ \hline
$k = 1$ & 133 & 114 & 37 \\
$k = 3$ & 124 & 93 & 30 \\
$k = 5$ & 114 & 77 & 28 \\
$k = 10$ & 105 & 70 & 28 \\
$k = 15$ & 98 & 62 & 23 \\
$k = 20$ & 94 & 61 & 24 \\ \hline
\end{tabular}
\end{table}

\begin{table}[t]
\small
\centering
\caption{\highlight{Average percentage of vertices that are not fully
		connected in the
maximum $k$-defective clique}}
\label{table:percentage_missing}
\begin{tabular}{c|ccc}
\hline
& Real-world & Facebook & DIMACS10\&SNAP \\ \hline
$k = 1$ & $19.2\%$ & $6.1\%$  & $16.9\%$ \\
$k = 3$ & $33.7\%$ & $15.9\%$ & $32.3\%$ \\
$k = 5$ & $43.3\%$ & $23.0\%$ & $46.6\%$ \\
$k = 10$ & $52.5\%$ & $34.4\%$ & $56.8\%$ \\
$k = 15$ & $59.5\%$ & $43.7\%$ & $64.7\%$ \\
$k = 20$ & $62.9\%$ & $50.3\%$ & $65.9\%$ \\ \hline
\end{tabular}
\end{table}

\highlight{
Secondly, we look into the actual maximum $k$-defective clique and check
(1)~whether it is an extension of a maximum clique and (2)~what fraction of
its vertices have missing neighbors. Note that, the maximum $k$-defective
clique in a graph is not unique, and the results reported here are based
on the maximum $k$-defective clique found by \kdc and thus are only for
the testings that finish within the time limit of 3 hours.
The results are shown in Tables~\ref{table:extension_of_clique} and
\ref{table:percentage_missing}, respectively.
From \ctab\ref{table:extension_of_clique}, we can see that for many,
although not all, of the graphs, the maximum $k$-defective
clique found by \kdc is an
extension of the maximum clique; specifically, the lowest fraction is
$\frac{62}{101} \approx 61.4\%$ which is achieved for $k=15$ on the
Facebook graphs collection. Nevertheless, we have shown in
\ctab\ref{table:size_ratio} that maximum $k$-defective cliques are larger
than maximum cliques.
From \ctab\ref{table:percentage_missing}, we can see that the percentage
of vertices that are not fully connected in a maximum $k$-defective
clique increases along with $k$, which is as expected. For $k\geq 10$ on
the real-world graphs collection,
more than half of the vertices in a maximum $k$-defective clique have
missing neighbors in the $k$-defective clique.
}

\section{Related Work}

The study of $k$-defective clique computation is still in its early
stage. The first exact algorithm for computing the maximum $k$-defective
clique was proposed in~\cite{COA13:Trukhanov}, which is based on the
Russian doll search~\cite{AAAI96:Verfaillie}, a solver for general
constraint optimization problems.
The algorithm of \cite{COA13:Trukhanov} was then improved
in~\cite{DAM18:Gschwind} with new preprocessing rules as well
as a better implementation.
A branch-and-price framework was designed in~\cite{IJC21:Gschwind}.
A continuous cubic formulation was established in~\cite{MP22:Stozhkov},
which generalizes the Motzkin-Straus formulation from the maximum clique
problem to the maximum $k$-defective clique problem; however, only
heuristic algorithms are designed in~\cite{MP22:Stozhkov}.
Chen et al.~\cite{COR21:Chen} proposed the \madec algorithm whose time
complexity beats the trivial $\bigo^*(2^n)$ time complexity, and
developed a graph coloring-based upper bound as well as other pruning
techniques.
The \kdbb algorithm proposed in~\cite{AAAI22:Gao} is the currently
fastest algorithm in practice, but its time complexity is the trivial
$\bigo^*(2^n)$. In this paper we proposed the \kdc algorithm which not
only has a better time complexity but also runs significantly faster in
practice than all existing algorithms.

The problem of (approximately) counting all $k$-defective cliques of a
particular size, for the special cases of $k=1$ and $2$, was recently
formulated and studied in~\cite{WWW20:Jain}. As the property of
$k$-defective clique is {\em hereditary}, the number of $k$-defective
cliques could explode drastically when the maximum $k$-defective clique
size increases. Thus, the maximum $k$-defective clique size may provide a
rough indication on the counting results. In addition, the pruning
techniques proposed in this paper may speed up the enumeration and counting
of {\em large} $k$-defective cliques.

Another related problem is maximum clique computation, which has been
extensively studied both theoretically and practically.
From a theoretical perspective, the worst case time complexity has
been gradually improved from
$\bigo^*(1.4422^n)$ to $\bigo^*(1.2599^n)$~\cite{SJC77:Tarjan},
$\bigo^*(1.2346^n)$~\cite{TC86:Jian}, and
$\bigo^*(1.2108^n)$~\cite{JA86:Robson}, with the state of the art being
$\bigo^*(1.1888^n)$~\cite{clique}; however, these algorithms are 
of theoretical interests only and have not been implemented. On the other
hand, a plethora of practical algorithms, without caring about the time
complexity analysis, have also been designed and implemented, 
\eg,~\cite{ORL90:Carraghan,COR17:Li,ICTAI13:Li,JGO94:Pardalos,IM15:Pattabiraman,JSC15:Rossi,COR16:Segundo,WALCOM17:Tomita,WALCOM10:Tomita,ICDE13:Xiang,KDD19:Chang}.
For these algorithms, upper bounds have been demonstrated to be critical
for the practical efficiency, and the most successful upper bounds are
based on graph coloring and MaxSAT reasoning.
However, these techniques cannot be easily extended to compute the maximum
$k$-defective clique for $k \geq 1$, despite that $k$-defective clique
is a relaxation of clique and $0$-defective cliques are just cliques.
For example, it was attempted in~\cite{COR21:Chen} to adapt the graph
coloring to compute an upper bound of the maximum $k$-defective clique
size, but as we demonstrated, the adaptation failed to compute a tight
upper bound and is not effective in improving the efficiency. In contrast,
we in this paper proposed a much tighter upper bound based on graph
coloring.

\section{Finding Top-$r$ $k$-defective Cliques}

In this section, we briefly discuss how to extend our techniques to two
variants of finding top-$r$ $k$-defective cliques. A thorough
investigation of these problems is beyond the scope of this paper, and
will be our future work.

Firstly, our techniques can be extended
to find top-$r$ maximal $k$-defective cliques, \ie, find the $r$ maximal
$k$-defective cliques that are largest. To do so, we will need to modify
our algorithm to enumerate all large maximal $k$-defective cliques.
Specifically, we will need to (1)~change the $d_g(u)$ condition of {\bf
RR2} to $d_g(u) \geq |V(g)|-1$, (2)~store in ${\cal C}$ the set of $r$
currently found largest maximal $k$-defective cliques rather than just the
single largest one, (3)~change the lower bound $\lb$ used in {\bf
RR3}--{\bf RR6} to be the size of
the smallest $k$-defective clique in ${\cal C}$. Due to the first
change, the time complexity would be $\bigo^*(\gamma_{2k}^n)$, the same
as the maximum $k$-defective clique computation algorithm
of~\cite{COR21:Chen}.

Secondly, our techniques can be extended to find top-$r$ diversified
$k$-defective cliques, \ie, find $r$ $k$-defective cliques that
collectively cover/contain the most number of distinct vertices.
Specifically, we iteratively conduct the following until $r$
$k$-defective cliques have been reported or the graph becomes empty:
\begin{enumerate}
	\item find the maximum $k$-defective clique $C$ in the current graph
		by invoking \kdc,
	\item remove $C$ from the current graph.
\end{enumerate}
Note that, this approach may not find the optimal result, but the
reported result provides a $(1-\frac{1}{e})$-approximation guarantee. The
time complexity is simply $r$ times that of \kdc.

\section{Conclusion}
\label{sec:conclusion}

In this paper, we advanced the state of the art for the problem of exact
maximum $k$-defective clique computation, in terms of both worst case
time complexity and practical performance. In specific, we first
developed a general framework \kdc based on our newly designed branching
rule {\bf BR} and reduction rules {\bf RR1} and {\bf RR2}. We proved that
our framework beats the trivial time complexity of $\bigo^*(2^n)$ and
achieves a better time complexity than all existing algorithms.
Then to make \kdc practically efficient, we further proposed a new upper
bound {\bf UB1}, two reduction rules {\bf RR3} and {\bf RR4}, as well as
an algorithm for efficiently computing a large initial solution.
Extensive empirical studies on \highlight{three} benchmark graph
collections with \highlight{$290$} graphs in total
demonstrated the practical superiority of \kdc over the existing
algorithms.

\begin{acks}
The author is supported by the Australian Research Council Fundings of
FT180100256 and DP220103731.
\end{acks}

\balance
{
\bibliographystyle{ACM-Reference-Format}
\bibliography{sigproc}


\begin{thebibliography}{50}


\ifx \showCODEN    \undefined \def \showCODEN     #1{\unskip}     \fi
\ifx \showDOI      \undefined \def \showDOI       #1{#1}\fi
\ifx \showISBNx    \undefined \def \showISBNx     #1{\unskip}     \fi
\ifx \showISBNxiii \undefined \def \showISBNxiii  #1{\unskip}     \fi
\ifx \showISSN     \undefined \def \showISSN      #1{\unskip}     \fi
\ifx \showLCCN     \undefined \def \showLCCN      #1{\unskip}     \fi
\ifx \shownote     \undefined \def \shownote      #1{#1}          \fi
\ifx \showarticletitle \undefined \def \showarticletitle #1{#1}   \fi
\ifx \showURL      \undefined \def \showURL       {\relax}        \fi
\providecommand\bibfield[2]{#2}
\providecommand\bibinfo[2]{#2}
\providecommand\natexlab[1]{#1}
\providecommand\showeprint[2][]{arXiv:#2}

\bibitem[Abello et~al\mbox{.}(2002)]%
        {LATIN02:Abello}
\bibfield{author}{\bibinfo{person}{James Abello}, \bibinfo{person}{Mauricio
  G.~C. Resende}, {and} \bibinfo{person}{Sandra Sudarsky}.}
  \bibinfo{year}{2002}\natexlab{}.
\newblock \showarticletitle{Massive Quasi-Clique Detection}. In
  \bibinfo{booktitle}{\emph{Proc. of {LATIN}'02}}
  \emph{(\bibinfo{series}{Lecture Notes in Computer Science},
  Vol.~\bibinfo{volume}{2286})}. \bibinfo{publisher}{Springer},
  \bibinfo{pages}{598--612}.
\newblock


\bibitem[Ahmed et~al\mbox{.}(2016)]%
        {ahmed2016survey}
\bibfield{author}{\bibinfo{person}{Mohiuddin Ahmed},
  \bibinfo{person}{Abdun~Naser Mahmood}, {and} \bibinfo{person}{Md~Rafiqul
  Islam}.} \bibinfo{year}{2016}\natexlab{}.
\newblock \showarticletitle{A survey of anomaly detection techniques in
  financial domain}.
\newblock \bibinfo{journal}{\emph{Future Generation Computer Systems}}
  \bibinfo{volume}{55} (\bibinfo{year}{2016}), \bibinfo{pages}{278--288}.
\newblock


\bibitem[Angel et~al\mbox{.}(2014)]%
        {VLDBJ14:Angel}
\bibfield{author}{\bibinfo{person}{Albert Angel}, \bibinfo{person}{Nick
  Koudas}, \bibinfo{person}{Nikos Sarkas}, \bibinfo{person}{Divesh Srivastava},
  \bibinfo{person}{Michael Svendsen}, {and} \bibinfo{person}{Srikanta
  Tirthapura}.} \bibinfo{year}{2014}\natexlab{}.
\newblock \showarticletitle{Dense subgraph maintenance under streaming edge
  weight updates for real-time story identification}.
\newblock \bibinfo{journal}{\emph{{VLDB} J.}} \bibinfo{volume}{23},
  \bibinfo{number}{2} (\bibinfo{year}{2014}), \bibinfo{pages}{175--199}.
\newblock


\bibitem[Balasundaram et~al\mbox{.}(2011)]%
        {OR11:Balasundaram}
\bibfield{author}{\bibinfo{person}{Balabhaskar Balasundaram},
  \bibinfo{person}{Sergiy Butenko}, {and} \bibinfo{person}{Illya~V. Hicks}.}
  \bibinfo{year}{2011}\natexlab{}.
\newblock \showarticletitle{Clique Relaxations in Social Network Analysis: The
  Maximum \emph{k}-Plex Problem}.
\newblock \bibinfo{journal}{\emph{Operations Research}} \bibinfo{volume}{59},
  \bibinfo{number}{1} (\bibinfo{year}{2011}), \bibinfo{pages}{133--142}.
\newblock


\bibitem[Bedi and Sharma(2016)]%
        {bedi2016community}
\bibfield{author}{\bibinfo{person}{Punam Bedi} {and} \bibinfo{person}{Chhavi
  Sharma}.} \bibinfo{year}{2016}\natexlab{}.
\newblock \showarticletitle{Community detection in social networks}.
\newblock \bibinfo{journal}{\emph{Wiley Interdisciplinary Reviews: Data Mining
  and Knowledge Discovery}} \bibinfo{volume}{6}, \bibinfo{number}{3}
  (\bibinfo{year}{2016}), \bibinfo{pages}{115--135}.
\newblock


\bibitem[Bourjolly et~al\mbox{.}(2002)]%
        {EJOR02:Bourjolly}
\bibfield{author}{\bibinfo{person}{Jean{-}Marie Bourjolly},
  \bibinfo{person}{Gilbert Laporte}, {and} \bibinfo{person}{Gilles Pesant}.}
  \bibinfo{year}{2002}\natexlab{}.
\newblock \showarticletitle{An exact algorithm for the maximum k-club problem
  in an undirected graph}.
\newblock \bibinfo{journal}{\emph{Eur. J. Oper. Res.}} \bibinfo{volume}{138},
  \bibinfo{number}{1} (\bibinfo{year}{2002}), \bibinfo{pages}{21--28}.
\newblock


\bibitem[Carraghan and Pardalos(1990)]%
        {ORL90:Carraghan}
\bibfield{author}{\bibinfo{person}{Randy Carraghan} {and}
  \bibinfo{person}{Panos~M. Pardalos}.} \bibinfo{year}{1990}\natexlab{}.
\newblock \showarticletitle{An Exact Algorithm for the Maximum Clique Problem}.
\newblock \bibinfo{journal}{\emph{Oper. Res. Lett.}} \bibinfo{volume}{9},
  \bibinfo{number}{6} (\bibinfo{date}{Nov.} \bibinfo{year}{1990}),
  \bibinfo{pages}{375--382}.
\newblock
\showISSN{0167-6377}


\bibitem[Chang(2019)]%
        {KDD19:Chang}
\bibfield{author}{\bibinfo{person}{Lijun Chang}.}
  \bibinfo{year}{2019}\natexlab{}.
\newblock \showarticletitle{Efficient Maximum Clique Computation over Large
  Sparse Graphs}. In \bibinfo{booktitle}{\emph{Proc. of KDD'19}}.
  \bibinfo{pages}{529--538}.
\newblock


\bibitem[Chang(2020)]%
        {VLDBJ20:Chang}
\bibfield{author}{\bibinfo{person}{Lijun Chang}.}
  \bibinfo{year}{2020}\natexlab{}.
\newblock \showarticletitle{Efficient maximum clique computation and
  enumeration over large sparse graphs}.
\newblock \bibinfo{journal}{\emph{{VLDB} J.}} \bibinfo{volume}{29},
  \bibinfo{number}{5} (\bibinfo{year}{2020}), \bibinfo{pages}{999--1022}.
\newblock


\bibitem[Chang and Qin(2018)]%
        {Book18:Chang}
\bibfield{author}{\bibinfo{person}{Lijun Chang} {and} \bibinfo{person}{Lu
  Qin}.} \bibinfo{year}{2018}\natexlab{}.
\newblock \bibinfo{booktitle}{\emph{Cohesive Subgraph Computation over Large
  Sparse Graphs}}.
\newblock \bibinfo{publisher}{Springer Series in the Data Sciences}.
\newblock


\bibitem[Chen et~al\mbox{.}(2021)]%
        {COR21:Chen}
\bibfield{author}{\bibinfo{person}{Xiaoyu Chen}, \bibinfo{person}{Yi Zhou},
  \bibinfo{person}{Jin{-}Kao Hao}, {and} \bibinfo{person}{Mingyu Xiao}.}
  \bibinfo{year}{2021}\natexlab{}.
\newblock \showarticletitle{Computing maximum k-defective cliques in massive
  graphs}.
\newblock \bibinfo{journal}{\emph{Comput. Oper. Res.}}  \bibinfo{volume}{127}
  (\bibinfo{year}{2021}), \bibinfo{pages}{105131}.
\newblock


\bibitem[Cormen et~al\mbox{.}(2001)]%
        {Book:CSRL}
\bibfield{author}{\bibinfo{person}{Thomas~H. Cormen},
  \bibinfo{person}{Charles~E. Leiserson}, \bibinfo{person}{Ronald~L. Rivest},
  {and} \bibinfo{person}{Clifford Stein}.} \bibinfo{year}{2001}\natexlab{}.
\newblock \bibinfo{booktitle}{\emph{Introduction to Algorithms}}.
\newblock \bibinfo{publisher}{McGraw-Hill Higher Education}.
\newblock
\showISBNx{0070131511}


\bibitem[Dai et~al\mbox{.}(2022)]%
        {CIKM22:Dai}
\bibfield{author}{\bibinfo{person}{Qiangqiang Dai}, \bibinfo{person}{Rong{-}Hua
  Li}, \bibinfo{person}{Hongchao Qin}, \bibinfo{person}{Meihao Liao}, {and}
  \bibinfo{person}{Guoren Wang}.} \bibinfo{year}{2022}\natexlab{}.
\newblock \showarticletitle{Scaling Up Maximal \emph{k}-plex Enumeration}. In
  \bibinfo{booktitle}{\emph{Proc. of CIKM'22}}. \bibinfo{pages}{345--354}.
\newblock


\bibitem[Eppstein et~al\mbox{.}(2013)]%
        {JEA13:Eppstein}
\bibfield{author}{\bibinfo{person}{David Eppstein}, \bibinfo{person}{Maarten
  L{\"{o}}ffler}, {and} \bibinfo{person}{Darren Strash}.}
  \bibinfo{year}{2013}\natexlab{}.
\newblock \showarticletitle{Listing All Maximal Cliques in Large Sparse
  Real-World Graphs}.
\newblock \bibinfo{journal}{\emph{{ACM} Journal of Experimental Algorithmics}}
  \bibinfo{volume}{18} (\bibinfo{year}{2013}).
\newblock


\bibitem[Fomin and Kratsch(2010)]%
        {Book2010:Fomin}
\bibfield{author}{\bibinfo{person}{Fedor~V. Fomin} {and}
  \bibinfo{person}{Dieter Kratsch}.} \bibinfo{year}{2010}\natexlab{}.
\newblock \bibinfo{booktitle}{\emph{Exact Exponential Algorithms}}.
\newblock \bibinfo{publisher}{Springer}.
\newblock


\bibitem[Gao et~al\mbox{.}(2022)]%
        {AAAI22:Gao}
\bibfield{author}{\bibinfo{person}{Jian Gao}, \bibinfo{person}{Zhenghang Xu},
  \bibinfo{person}{Ruizhi Li}, {and} \bibinfo{person}{Minghao Yin}.}
  \bibinfo{year}{2022}\natexlab{}.
\newblock \showarticletitle{An Exact Algorithm with New Upper Bounds for the
  Maximum k-Defective Clique Problem in Massive Sparse Graphs}. In
  \bibinfo{booktitle}{\emph{Proc. of AAAI'22}}. \bibinfo{pages}{10174--10183}.
\newblock


\bibitem[Gschwind et~al\mbox{.}(2021)]%
        {IJC21:Gschwind}
\bibfield{author}{\bibinfo{person}{Timo Gschwind}, \bibinfo{person}{Stefan
  Irnich}, \bibinfo{person}{Fabio Furini}, {and}
  \bibinfo{person}{Roberto~Wolfler Calvo}.} \bibinfo{year}{2021}\natexlab{}.
\newblock \showarticletitle{A Branch-and-Price Framework for Decomposing Graphs
  into Relaxed Cliques}.
\newblock \bibinfo{journal}{\emph{{INFORMS} J. Comput.}} \bibinfo{volume}{33},
  \bibinfo{number}{3} (\bibinfo{year}{2021}), \bibinfo{pages}{1070--1090}.
\newblock


\bibitem[Gschwind et~al\mbox{.}(2018)]%
        {DAM18:Gschwind}
\bibfield{author}{\bibinfo{person}{Timo Gschwind}, \bibinfo{person}{Stefan
  Irnich}, {and} \bibinfo{person}{Isabel Podlinski}.}
  \bibinfo{year}{2018}\natexlab{}.
\newblock \showarticletitle{Maximum weight relaxed cliques and Russian Doll
  Search revisited}.
\newblock \bibinfo{journal}{\emph{Discret. Appl. Math.}}  \bibinfo{volume}{234}
  (\bibinfo{year}{2018}), \bibinfo{pages}{131--138}.
\newblock


\bibitem[H{\aa}stad(1996)]%
        {FOCS96:Hastad}
\bibfield{author}{\bibinfo{person}{Johan H{\aa}stad}.}
  \bibinfo{year}{1996}\natexlab{}.
\newblock \showarticletitle{Clique is Hard to Approximate Within
  n\({}^{\mbox{1-epsilon}}\)}. In \bibinfo{booktitle}{\emph{Proc. of FOCS'96}}.
  \bibinfo{pages}{627--636}.
\newblock


\bibitem[Jain and Seshadhri(2020a)]%
        {WSDM20:Jain}
\bibfield{author}{\bibinfo{person}{Shweta Jain} {and} \bibinfo{person}{C.
  Seshadhri}.} \bibinfo{year}{2020}\natexlab{a}.
\newblock \showarticletitle{The Power of Pivoting for Exact Clique Counting}.
  In \bibinfo{booktitle}{\emph{Proc. {WSDM}'20}}. \bibinfo{publisher}{{ACM}},
  \bibinfo{pages}{268--276}.
\newblock


\bibitem[Jain and Seshadhri(2020b)]%
        {WWW20:Jain}
\bibfield{author}{\bibinfo{person}{Shweta Jain} {and} \bibinfo{person}{C.
  Seshadhri}.} \bibinfo{year}{2020}\natexlab{b}.
\newblock \showarticletitle{Provably and Efficiently Approximating Near-cliques
  using the Tur{\'{a}}n Shadow: {PEANUTS}}. In \bibinfo{booktitle}{\emph{Proc.
  of {WWW}'20}}. \bibinfo{publisher}{{ACM} / {IW3C2}},
  \bibinfo{pages}{1966--1976}.
\newblock


\bibitem[Jian(1986)]%
        {TC86:Jian}
\bibfield{author}{\bibinfo{person}{Tang Jian}.}
  \bibinfo{year}{1986}\natexlab{}.
\newblock \showarticletitle{An \emph{O}(2\({}^{\mbox{0.304\emph{n}}}\))
  Algorithm for Solving Maximum Independent Set Problem}.
\newblock \bibinfo{journal}{\emph{{IEEE} Trans. Computers}}
  \bibinfo{volume}{35}, \bibinfo{number}{9} (\bibinfo{year}{1986}),
  \bibinfo{pages}{847--851}.
\newblock


\bibitem[Karp(1972)]%
        {CCC72:Karp}
\bibfield{author}{\bibinfo{person}{Richard~M. Karp}.}
  \bibinfo{year}{1972}\natexlab{}.
\newblock \showarticletitle{Reducibility Among Combinatorial Problems}. In
  \bibinfo{booktitle}{\emph{Proc. of CCC'72}}. \bibinfo{pages}{85--103}.
\newblock


\bibitem[Lee et~al\mbox{.}(2010)]%
        {MMGD10:Lee}
\bibfield{author}{\bibinfo{person}{Victor~E. Lee}, \bibinfo{person}{Ning Ruan},
  \bibinfo{person}{Ruoming Jin}, {and} \bibinfo{person}{Charu~C. Aggarwal}.}
  \bibinfo{year}{2010}\natexlab{}.
\newblock \showarticletitle{A Survey of Algorithms for Dense Subgraph
  Discovery}.
\newblock In \bibinfo{booktitle}{\emph{Managing and Mining Graph Data}}.
  \bibinfo{series}{Advances in Database Systems}, Vol.~\bibinfo{volume}{40}.
  \bibinfo{publisher}{Springer}, \bibinfo{pages}{303--336}.
\newblock


\bibitem[Li et~al\mbox{.}(2013)]%
        {ICTAI13:Li}
\bibfield{author}{\bibinfo{person}{Chu{-}Min Li}, \bibinfo{person}{Zhiwen
  Fang}, {and} \bibinfo{person}{Ke Xu}.} \bibinfo{year}{2013}\natexlab{}.
\newblock \showarticletitle{Combining MaxSAT Reasoning and Incremental Upper
  Bound for the Maximum Clique Problem}. In \bibinfo{booktitle}{\emph{Proc. of
  ICTAI'13}}.
\newblock


\bibitem[Li et~al\mbox{.}(2017)]%
        {COR17:Li}
\bibfield{author}{\bibinfo{person}{Chu{-}Min Li}, \bibinfo{person}{Hua Jiang},
  {and} \bibinfo{person}{Felip Many{\`{a}}}.} \bibinfo{year}{2017}\natexlab{}.
\newblock \showarticletitle{On minimization of the number of branches in
  branch-and-bound algorithms for the maximum clique problem}.
\newblock \bibinfo{journal}{\emph{Computers {\&} {OR}}}  \bibinfo{volume}{84}
  (\bibinfo{year}{2017}), \bibinfo{pages}{1--15}.
\newblock


\bibitem[Li et~al\mbox{.}(2020)]%
        {PVLDB20:Li}
\bibfield{author}{\bibinfo{person}{Ronghua Li}, \bibinfo{person}{Sen Gao},
  \bibinfo{person}{Lu Qin}, \bibinfo{person}{Guoren Wang},
  \bibinfo{person}{Weihua Yang}, {and} \bibinfo{person}{Jeffrey~Xu Yu}.}
  \bibinfo{year}{2020}\natexlab{}.
\newblock \showarticletitle{Ordering Heuristics for k-clique Listing}.
\newblock \bibinfo{journal}{\emph{Proc. {VLDB} Endow.}} \bibinfo{volume}{13},
  \bibinfo{number}{11} (\bibinfo{year}{2020}), \bibinfo{pages}{2536--2548}.
\newblock


\bibitem[Matula and Beck(1983)]%
        {JACM83:Matula}
\bibfield{author}{\bibinfo{person}{David~W. Matula} {and}
  \bibinfo{person}{Leland~L. Beck}.} \bibinfo{year}{1983}\natexlab{}.
\newblock \showarticletitle{Smallest-Last Ordering and clustering and Graph
  Coloring Algorithms}.
\newblock \bibinfo{journal}{\emph{J. {ACM}}} \bibinfo{volume}{30},
  \bibinfo{number}{3} (\bibinfo{year}{1983}), \bibinfo{pages}{417--427}.
\newblock


\bibitem[Pardalos and Xue(1994)]%
        {JGO94:Pardalos}
\bibfield{author}{\bibinfo{person}{Panos~M. Pardalos} {and}
  \bibinfo{person}{Jue Xue}.} \bibinfo{year}{1994}\natexlab{}.
\newblock \showarticletitle{The maximum clique problem}.
\newblock \bibinfo{journal}{\emph{J. global Optimization}} \bibinfo{volume}{4},
  \bibinfo{number}{3} (\bibinfo{year}{1994}), \bibinfo{pages}{301--328}.
\newblock


\bibitem[Pattabiraman et~al\mbox{.}(2015)]%
        {IM15:Pattabiraman}
\bibfield{author}{\bibinfo{person}{Bharath Pattabiraman}, \bibinfo{person}{Md.
  Mostofa~Ali Patwary}, \bibinfo{person}{Assefaw~Hadish Gebremedhin},
  \bibinfo{person}{Wei{-}keng Liao}, {and} \bibinfo{person}{Alok~N.
  Choudhary}.} \bibinfo{year}{2015}\natexlab{}.
\newblock \showarticletitle{Fast Algorithms for the Maximum Clique Problem on
  Massive Graphs with Applications to Overlapping Community Detection}.
\newblock \bibinfo{journal}{\emph{Internet Mathematics}} \bibinfo{volume}{11},
  \bibinfo{number}{4-5} (\bibinfo{year}{2015}), \bibinfo{pages}{421--448}.
\newblock


\bibitem[Pattillo et~al\mbox{.}(2013)]%
        {EOR13:Pattillo}
\bibfield{author}{\bibinfo{person}{Jeffrey Pattillo}, \bibinfo{person}{Nataly
  Youssef}, {and} \bibinfo{person}{Sergiy Butenko}.}
  \bibinfo{year}{2013}\natexlab{}.
\newblock \showarticletitle{On clique relaxation models in network analysis}.
\newblock \bibinfo{journal}{\emph{Eur. J. Oper. Res.}} \bibinfo{volume}{226},
  \bibinfo{number}{1} (\bibinfo{year}{2013}), \bibinfo{pages}{9--18}.
\newblock


\bibitem[Robson(1986)]%
        {JA86:Robson}
\bibfield{author}{\bibinfo{person}{J.~M. Robson}.}
  \bibinfo{year}{1986}\natexlab{}.
\newblock \showarticletitle{Algorithms for Maximum Independent Sets}.
\newblock \bibinfo{journal}{\emph{J. Algorithms}} \bibinfo{volume}{7},
  \bibinfo{number}{3} (\bibinfo{year}{1986}), \bibinfo{pages}{425--440}.
\newblock


\bibitem[Robson(2001)]%
        {clique}
\bibfield{author}{\bibinfo{person}{J.~M. Robson}.}
  \bibinfo{year}{2001}\natexlab{}.
\newblock \bibinfo{title}{Finding a maximum independent set in time
  $O(2^{n/4})$}.
\newblock
  \bibinfo{howpublished}{\small\url{https://www.labri.fr/perso/robson/mis/techrep.html}}.
\newblock


\bibitem[Rossi et~al\mbox{.}(2015)]%
        {JSC15:Rossi}
\bibfield{author}{\bibinfo{person}{Ryan~A. Rossi}, \bibinfo{person}{David~F.
  Gleich}, {and} \bibinfo{person}{Assefaw~Hadish Gebremedhin}.}
  \bibinfo{year}{2015}\natexlab{}.
\newblock \showarticletitle{Parallel Maximum Clique Algorithms with
  Applications to Network Analysis}.
\newblock \bibinfo{journal}{\emph{{SIAM} J. Scientific Computing}}
  \bibinfo{volume}{37}, \bibinfo{number}{5} (\bibinfo{year}{2015}).
\newblock


\bibitem[Sachs(1963)]%
        {JLMS63:Sachs}
\bibfield{author}{\bibinfo{person}{H. Sachs}.} \bibinfo{year}{1963}\natexlab{}.
\newblock \showarticletitle{{Regular Graphs with Given Girth and Restricted
  Circuits}}.
\newblock \bibinfo{journal}{\emph{Journal of the London Mathematical Society}}
  \bibinfo{volume}{s1-38}, \bibinfo{number}{1} (\bibinfo{year}{1963}),
  \bibinfo{pages}{423--429}.
\newblock


\bibitem[Segundo et~al\mbox{.}(2016)]%
        {COR16:Segundo}
\bibfield{author}{\bibinfo{person}{Pablo~San Segundo}, \bibinfo{person}{Alvaro
  Lopez}, {and} \bibinfo{person}{Panos~M. Pardalos}.}
  \bibinfo{year}{2016}\natexlab{}.
\newblock \showarticletitle{A new exact maximum clique algorithm for large and
  massive sparse graphs}.
\newblock \bibinfo{journal}{\emph{Computers \& Operations Research}}
  \bibinfo{volume}{66} (\bibinfo{year}{2016}), \bibinfo{pages}{81--94}.
\newblock


\bibitem[Seidman(1983)]%
        {SN83:Seidman}
\bibfield{author}{\bibinfo{person}{Stephen~B. Seidman}.}
  \bibinfo{year}{1983}\natexlab{}.
\newblock \showarticletitle{Network structure and minimum degree}.
\newblock \bibinfo{journal}{\emph{Social Networks}} \bibinfo{volume}{5},
  \bibinfo{number}{3} (\bibinfo{year}{1983}), \bibinfo{pages}{269 -- 287}.
\newblock
\showISSN{0378-8733}


\bibitem[Sherali et~al\mbox{.}(2002)]%
        {TS02:Sherali}
\bibfield{author}{\bibinfo{person}{Hanif~D. Sherali}, \bibinfo{person}{J.~Cole
  Smith}, {and} \bibinfo{person}{Antonio~A. Trani}.}
  \bibinfo{year}{2002}\natexlab{}.
\newblock \showarticletitle{An Airspace Planning Model for Selecting
  Flight-plans Under Workload, Safety, and Equity Considerations}.
\newblock \bibinfo{journal}{\emph{Transp. Sci.}} \bibinfo{volume}{36},
  \bibinfo{number}{4} (\bibinfo{year}{2002}), \bibinfo{pages}{378--397}.
\newblock


\bibitem[Stozhkov et~al\mbox{.}(2022)]%
        {MP22:Stozhkov}
\bibfield{author}{\bibinfo{person}{Vladimir Stozhkov}, \bibinfo{person}{Austin
  Buchanan}, \bibinfo{person}{Sergiy Butenko}, {and} \bibinfo{person}{Vladimir
  Boginski}.} \bibinfo{year}{2022}\natexlab{}.
\newblock \showarticletitle{Continuous cubic formulations for cluster detection
  problems in networks}.
\newblock \bibinfo{journal}{\emph{Math. Program.}} \bibinfo{volume}{196},
  \bibinfo{number}{1} (\bibinfo{year}{2022}), \bibinfo{pages}{279--307}.
\newblock


\bibitem[Suratanee et~al\mbox{.}(2014)]%
        {suratanee2014characterizing}
\bibfield{author}{\bibinfo{person}{Apichat Suratanee},
  \bibinfo{person}{Martin~H Schaefer}, \bibinfo{person}{Matthew~J Betts},
  \bibinfo{person}{Zita Soons}, \bibinfo{person}{Heiko Mannsperger},
  \bibinfo{person}{Nathalie Harder}, \bibinfo{person}{Marcus Oswald},
  \bibinfo{person}{Markus Gipp}, \bibinfo{person}{Ellen Ramminger},
  \bibinfo{person}{Guillermo Marcus}, {et~al\mbox{.}}}
  \bibinfo{year}{2014}\natexlab{}.
\newblock \showarticletitle{Characterizing protein interactions employing a
  genome-wide siRNA cellular phenotyping screen}.
\newblock \bibinfo{journal}{\emph{PLoS computational biology}}
  \bibinfo{volume}{10}, \bibinfo{number}{9} (\bibinfo{year}{2014}),
  \bibinfo{pages}{e1003814}.
\newblock


\bibitem[Tarjan and Trojanowski(1977)]%
        {SJC77:Tarjan}
\bibfield{author}{\bibinfo{person}{Robert~Endre Tarjan} {and}
  \bibinfo{person}{Anthony~E. Trojanowski}.} \bibinfo{year}{1977}\natexlab{}.
\newblock \showarticletitle{Finding a Maximum Independent Set}.
\newblock \bibinfo{journal}{\emph{{SIAM} J. Comput.}} \bibinfo{volume}{6},
  \bibinfo{number}{3} (\bibinfo{year}{1977}), \bibinfo{pages}{537--546}.
\newblock


\bibitem[Tomita(2017)]%
        {WALCOM17:Tomita}
\bibfield{author}{\bibinfo{person}{Etsuji Tomita}.}
  \bibinfo{year}{2017}\natexlab{}.
\newblock \showarticletitle{Efficient Algorithms for Finding Maximum and
  Maximal Cliques and Their Applications}. In \bibinfo{booktitle}{\emph{Proc.
  of WALCOM'17}}. \bibinfo{pages}{3--15}.
\newblock


\bibitem[Tomita et~al\mbox{.}(2010)]%
        {WALCOM10:Tomita}
\bibfield{author}{\bibinfo{person}{Etsuji Tomita}, \bibinfo{person}{Yoichi
  Sutani}, \bibinfo{person}{Takanori Higashi}, \bibinfo{person}{Shinya
  Takahashi}, {and} \bibinfo{person}{Mitsuo Wakatsuki}.}
  \bibinfo{year}{2010}\natexlab{}.
\newblock \showarticletitle{A simple and faster branch-and-bound algorithm for
  finding a maximum clique}. In \bibinfo{booktitle}{\emph{Proc. of WALCOM'10}}.
  \bibinfo{pages}{191--203}.
\newblock


\bibitem[Trukhanov et~al\mbox{.}(2013)]%
        {COA13:Trukhanov}
\bibfield{author}{\bibinfo{person}{Svyatoslav Trukhanov},
  \bibinfo{person}{Chitra Balasubramaniam}, \bibinfo{person}{Balabhaskar
  Balasundaram}, {and} \bibinfo{person}{Sergiy Butenko}.}
  \bibinfo{year}{2013}\natexlab{}.
\newblock \showarticletitle{Algorithms for detecting optimal hereditary
  structures in graphs, with application to clique relaxations}.
\newblock \bibinfo{journal}{\emph{Comput. Optim. Appl.}} \bibinfo{volume}{56},
  \bibinfo{number}{1} (\bibinfo{year}{2013}), \bibinfo{pages}{113--130}.
\newblock


\bibitem[Verfaillie et~al\mbox{.}(1996)]%
        {AAAI96:Verfaillie}
\bibfield{author}{\bibinfo{person}{G{\'{e}}rard Verfaillie},
  \bibinfo{person}{Michel Lema{\^{\i}}tre}, {and} \bibinfo{person}{Thomas
  Schiex}.} \bibinfo{year}{1996}\natexlab{}.
\newblock \showarticletitle{Russian Doll Search for Solving Constraint
  Optimization Problems}. In \bibinfo{booktitle}{\emph{Proc. of AAAI'96}}.
  \bibinfo{publisher}{{AAAI} Press / The {MIT} Press},
  \bibinfo{pages}{181--187}.
\newblock


\bibitem[Wang and Cheng(2012)]%
        {PVLDB12:Wang}
\bibfield{author}{\bibinfo{person}{Jia Wang} {and} \bibinfo{person}{James
  Cheng}.} \bibinfo{year}{2012}\natexlab{}.
\newblock \showarticletitle{Truss Decomposition in Massive Networks}.
\newblock \bibinfo{journal}{\emph{PVLDB}} \bibinfo{volume}{5},
  \bibinfo{number}{9} (\bibinfo{year}{2012}).
\newblock


\bibitem[Xiang et~al\mbox{.}(2013)]%
        {ICDE13:Xiang}
\bibfield{author}{\bibinfo{person}{Jingen Xiang}, \bibinfo{person}{Cong Guo},
  {and} \bibinfo{person}{Ashraf Aboulnaga}.} \bibinfo{year}{2013}\natexlab{}.
\newblock \showarticletitle{Scalable maximum clique computation using
  mapreduce}. In \bibinfo{booktitle}{\emph{Proc. of ICDE'13}}.
  \bibinfo{pages}{74--85}.
\newblock


\bibitem[Yannakakis(1978)]%
        {STOC78:Yannakakis}
\bibfield{author}{\bibinfo{person}{Mihalis Yannakakis}.}
  \bibinfo{year}{1978}\natexlab{}.
\newblock \showarticletitle{Node- and Edge-Deletion NP-Complete Problems}. In
  \bibinfo{booktitle}{\emph{Proc. of STOC'78}}. \bibinfo{publisher}{{ACM}},
  \bibinfo{pages}{253--264}.
\newblock


\bibitem[Yu et~al\mbox{.}(2006)]%
        {Bio06:Yu}
\bibfield{author}{\bibinfo{person}{Haiyuan Yu}, \bibinfo{person}{Alberto
  Paccanaro}, \bibinfo{person}{Valery Trifonov}, {and} \bibinfo{person}{Mark
  Gerstein}.} \bibinfo{year}{2006}\natexlab{}.
\newblock \showarticletitle{Predicting interactions in protein networks by
  completing defective cliques}.
\newblock \bibinfo{journal}{\emph{Bioinform.}} \bibinfo{volume}{22},
  \bibinfo{number}{7} (\bibinfo{year}{2006}), \bibinfo{pages}{823--829}.
\newblock


\bibitem[Zhou et~al\mbox{.}(2020)]%
        {AAAI20:Zhou}
\bibfield{author}{\bibinfo{person}{Yi Zhou}, \bibinfo{person}{Jingwei Xu},
  \bibinfo{person}{Zhenyu Guo}, \bibinfo{person}{Mingyu Xiao}, {and}
  \bibinfo{person}{Yan Jin}.} \bibinfo{year}{2020}\natexlab{}.
\newblock \showarticletitle{Enumerating Maximal \emph{k}-Plexes with Worst-Case
  Time Guarantee}. In \bibinfo{booktitle}{\emph{Proc. of AAAI'20}}.
  \bibinfo{pages}{2442--2449}.
\newblock


\end{thebibliography}
}

\end{document}